\def\HideComments{TRUE}%

\ifx\SoCGVer\undefined
   \newcommand{\InSoCGVer}[1]{}%
   \newcommand{\InNotSoCGVer}[1]{#1}%
\else
   \newcommand{\InSoCGVer}[1]{#1}%
   \newcommand{\InNotSoCGVer}[1]{}%
\fi

\newcommand{\CiteFullVer}{\cite{rs-mvrms-14}}%

\InSoCGVer{%
   \documentclass[a4paper,UKenglish]{lipics}
}
\InNotSoCGVer{%
  \documentclass[11pt]{article}
  \usepackage[margin=2.54cm]{geometry}
}

\InNotSoCGVer{
   \newcommand{\myparagraph}[1]{\paragraph{#1}}
}
\InSoCGVer{
   \newcommand{\myparagraph}[1]{\smallskip\noindent{\textbf{#1}}}
}

\usepackage{amsfonts,amsmath,amssymb,amsthm,mathtools}
\usepackage{paralist}
\usepackage{algorithmic,algorithm}
\usepackage{bm}
\usepackage{xspace}
\usepackage{centernot}
\usepackage{fancybox}
\usepackage{framed}
\usepackage{xspace}
\usepackage{prettyref}
\usepackage{color}
\usepackage{graphics}
\usepackage{hyperref}%
\hypersetup{%
   breaklinks,%
   colorlinks=true,%
   urlcolor=[rgb]{0.2,0.0,0.0},
   linkcolor=[rgb]{0.5,0.0,0.0},%
   citecolor=[rgb]{0,0,0.445},
   filecolor=[rgb]{0,0,0.4},
   anchorcolor=[rgb]{0.0,0.1,0.2}
}

\newcommand{\ignore}[1]{}

\InNotSoCGVer{
 \newtheorem{theorem}{Theorem}[section]
}

\newtheorem{claim}[theorem]{Claim}

\InNotSoCGVer{
 \newtheorem{lemma}[theorem]{Lemma}
 \newtheorem{corollary}[theorem]{Corollary}
 
 \newtheorem{definition}[theorem]{Definition}
}

\theoremstyle{definition}
\newtheorem{Remark}[theorem]{Remark}
\newtheorem{observation}[theorem]{Observation}

\newcommand{\bF}{{\bf F}}

\newcommand{\cA}{{\cal A}}

\newcommand{\cC}{{\cal C}}

\newcommand{\cE}{{\cal E}}

\newcommand{\cH}{{\cal H}}
\newcommand{\cJ}{{\cal J}}
\newcommand{\cL}{{\cal L}}

\newcommand{\cS}{{\cal S}}
\newcommand{\cU}{{\cal U}}
\newcommand{\cV}{{\cal V}}

\newcommand{\LL}{\mathbb{L}}
\newcommand{\MM}{\mathbb{M}}
\newcommand{\NN}{\mathbb{N}}
\newcommand{\PP}{\mathbb{P}}
\newcommand{\RR}{\mathbb{R}}

\newcommand{\eps}{\varepsilon}

\newcommand{\Sec}[1]{\hyperref[sec:#1]{\S\ref*{sec:#1}}} 
\newcommand{\Eqn}[1]{\hyperref[eqn:#1]{(\ref*{eqn:#1})}} 
\newcommand{\Clm}[1]{\hyperref[clm:#1]{Claim~\ref*{clm:#1}}} 
\newcommand{\Fig}[1]{\hyperref[fig:#1]{Figure~\ref*{fig:#1}}} 
\newcommand{\Tab}[1]{\hyperref[tab:#1]{Table~\ref*{tab:#1}}} 
\newcommand{\Thm}[1]{\hyperref[thm:#1]{Theorem~\ref*{thm:#1}}} 
\newcommand{\Lem}[1]{\hyperref[lem:#1]{Lemma~\ref*{lem:#1}}} 
\newcommand{\Prop}[1]{\hyperref[prop:#1]{Proposition~\ref*{prop:#1}}} 
\newcommand{\Cor}[1]{\hyperref[cor:#1]{Corollary~\ref*{cor:#1}}} 
\newcommand{\Def}[1]{\hyperref[def:#1]{Definition~\ref*{def:#1}}} 
\newcommand{\Alg}[1]{\hyperref[alg:#1]{Algorithm~\ref*{alg:#1}}} 
\newcommand{\Ex}[1]{\hyperref[ex:#1]{Example~\ref*{ex:#1}}} 
\newcommand{\Obs}[1]{\hyperref[obs:#1]{Observation~\ref*{obs:#1}}} 

\DefineNamedColor{named}{RedViolet} {cmyk}{0.07,0.90,0,0.34}

\newcommand{\etal}{\textit{et~al.}\xspace}

\newcommand{\build}{{\tt build}}
\newcommand{\col}{col}

\newcommand{\cur}{cur}
\newcommand{\cut}{{\tt cut}}

\newcommand{\cost}{\mathop{cost}}

\newcommand{\h}{att}
\newcommand{\init}{{\tt init}}
\newcommand{\jc}{\cJ_C}
\newcommand{\lift}{{\tt lift}}
\newcommand{\mcol}{mcol}
\newcommand{\merge}{{\tt merge}}

\newcommand{\pal}{P}
\newcommand{\pmax}{P_{\max}}
\newcommand{\rain}{{\tt rain}}
\newcommand{\reeb}{\cC}
\newcommand{\redH}{\widetilde{H}}
\newcommand{\rep}{rep}
\newcommand{\stack}{K}
\newcommand{\surgery}{{\tt surgery}}
\newcommand{\touch}{T}
\newcommand{\update}{{\tt update}}
\newcommand{\wet}{{\tt wet}}

\newcommand{\pathTree}{P_{\mathcal{S}}}

\newcommand{\XSays}[2]{{
      {$\rule[-0.12cm]{0.2in}{0.5cm}$\fbox{\tt
            #1:} }
      \textcolor{red}{#2}
      \marginpar{\textcolor{blue}{#1}}
      {$\rule[0.1cm]{0.3in}{0.1cm}$\fbox{\tt
            end}$\rule[0.1cm]{0.3in}{0.1cm}$}
      }
   }
\ifx\HideComments\undefined%
\else%
\renewcommand{\XSays}[2]{}%
\fi%
   
\newcommand{\Ben}[1]{{\XSays{Ben}{#1}}}
\newcommand{\remove}[1]{}
\newcommand{\pth}[2][\!]{#1\left({#2}\right)}

\newcommand{\si}[1]{#1}

\begin{document}

\InNotSoCGVer{
\author{
  Benjamin Raichel
  \and
  C. Seshadhri
}

\title{
Avoiding the Global Sort: 
\break A Faster Contour Tree Algorithm%
\footnote{The full updated version of this paper is available on the arXiv \cite{rs-mvrms-14}}%
}
\date{}
}

\InSoCGVer{
\title{
Avoiding the Global Sort: 
\break A Faster Contour Tree Algorithm%
\footnote{The full updated version of this paper is available on the arXiv \cite{rs-mvrms-14}}%
}
\titlerunning{Avoiding the Global Sort: A Faster Contour Tree Algorithm} 

\author[1]{Benjamin Raichel}
\author[2]{C. Seshadhri}
  \affil[1]{%
      Department of Computer Science, %
      University of Texas at Dallas\\ %
      Richardson, TX, 75080, USA\\ %
      \texttt{\si{benjamin.raichel}@\si{utdallas}.\si{edu}} %
   }%
   \affil[2]{%
      Department of Computer Science, %
      University of California, Santa Cruz\\ %
      Santa Cruz, CA, 95064, USA\\ %
      \texttt{\si{scomandu}@\si{ucsc}.\si{edu}} %
   }%

\authorrunning{B. Raichel and C. Seshadhri} 

\Copyright{Benjamin Raichel and Seshadhri Comandur}


}

\maketitle
\InNotSoCGVer{
\thispagestyle{empty}
}

\begin{abstract}
We revisit the classical problem of computing the \emph{contour tree}
of a scalar field $f:\MM \to \RR$, where $\MM$ is a triangulated simplicial mesh in $\RR^d$. 
The contour tree is a fundamental topological structure that tracks
the evolution of level sets of $f$ and 
has numerous applications in data analysis and visualization.

All existing algorithms begin with a global sort of at least all critical values of $f$,
which can require (roughly) $\Omega(n\log n)$ time.
Existing lower bounds show that there are pathological instances where this sort is required.
We present the first algorithm whose time complexity depends
on the contour tree structure, and avoids the global sort for non-pathological inputs.
If $C$ denotes the set of critical points in $\MM$, the running time is
roughly $O(\sum_{v \in C} \log \ell_v)$, where $\ell_v$ is the depth of $v$ in the contour tree.
This matches all existing upper bounds, but is a significant improvement when the contour tree is short and fat.
Specifically, our approach ensures that any comparison made is between nodes in the same descending path in the contour tree,
allowing us to argue strong optimality properties of our algorithm.

Our algorithm requires several novel ideas: partitioning $\MM$ in well-behaved portions, 
a local growing procedure to iteratively build contour trees, and the use of heavy path 
decompositions for the time complexity analysis. 
\end{abstract}

\InNotSoCGVer{
\newpage
\pagenumbering{arabic}
}

\section{Introduction}

Geometric data is often represented as a function $f: \RR^d \to \RR$. Typically, a finite representation is given
by considering $f$ to be piecewise linear over some triangulated mesh (i.e.\ simplicial complex) $\MM$ in $\RR^d$.
\emph{Contour trees} are a topological structure used to represent and visualize
the function $f$. 
%
It is convenient to think of $f$ as a manifold sitting in $\RR^{d+1}$,
with the last coordinate (i.e.\ height) given by $f$. Imagine sweeping the hyperplane $x_{d+1} = h$
with $h$ going from $+\infty$ to $-\infty$. At every instance, the intersection of this plane 
with $f$ gives a set of connected components, the \emph{contours} at height $h$. As the sweeping 
proceeds various events occur: new contours are created or destroyed, contours merge into each other or
split into new components, contours acquire or lose handles. 
The contour tree is a concise representation of all these events.
Throughout we follow the definition of contour trees from \cite{kobps-ctsssit-97} which includes all changes in topology.  
For $d>2$, some subsequent works, such as \cite{csa-cctad-00}, only include changes in the number of components.

If $f$ is smooth, all points where the gradient of $f$ is zero are \emph{critical points}.
These points are the ``events" where the contour topology changes and form the vertices of the contour tree. 
An edge of the contour tree connects two critical points if one event immediately ``follows"
the other as the sweep plane makes its pass. (We provide formal definitions later.) 
\Fig{relative} and \Fig{sorting} 
show examples of simplicial complexes, with heights and their contour trees.
Think of the contour tree edges as pointing downwards. Leaves are either maxima or minima,
and internal nodes are either ``joins" or ``splits".


Consider $f: \MM \to \RR$, where $\MM$ is a triangulated mesh with $n$ vertices, $N$ faces in total, and $t \leq n$ critical points. 
(We assume that $f:\MM \to \RR$ a linear interpolant over distinct valued vertices, where the contour tree
$T$ has maximum degree $3$. The degree assumption simplifies the presentation, and is commonly made~\cite{kobps-ctsssit-97}.)
A fundamental result in this area is the algorithm of Carr, Snoeyink, and Axen to compute contour trees,
which runs in $O(n\log n + N\alpha(N))$ time~\cite{csa-cctad-00} (where $\alpha(\cdot)$ denotes the inverse
Ackermann function). In practical applications, $N$ is typically $\Theta(n)$
(certainly true for $d=2$). The most expensive operation is an initial sort of all the vertex
heights. Chiang \etal build on this approach to get a faster algorithm that only sorts the critical vertices,
yielding a running time of $O(t\log t + N)$~\cite{cllr-sooscctmp-05}.
Common applications for contour trees involve turbulent combustion or noisy data, where the number of critical points is likely to be $\Omega(n)$.
There is a worst-case
lower bound of $\Omega(t\log t)$ by Chiang \etal \cite{cllr-sooscctmp-05}, based on a construction of Bajaj \etal \cite{BaKr+98}.

All previous algorithms begin by sorting (at least) the critical points. Can we beat this sorting bound for certain instances,
and can we characterize which inputs are hard? Intuitively, points that are incomparable in the contour tree
do not need to be compared. Look at \Fig{relative} to see such an example. All previous algorithms waste time
sorting all the maxima. Also consider the surface of \Fig{sorting}. The final contour tree is basically two binary
trees joined at their roots, and we do not need the entire sorted order of critical points
to construct the contour tree. 

\begin{figure}[h]\centering
    \includegraphics[width=.26\linewidth]{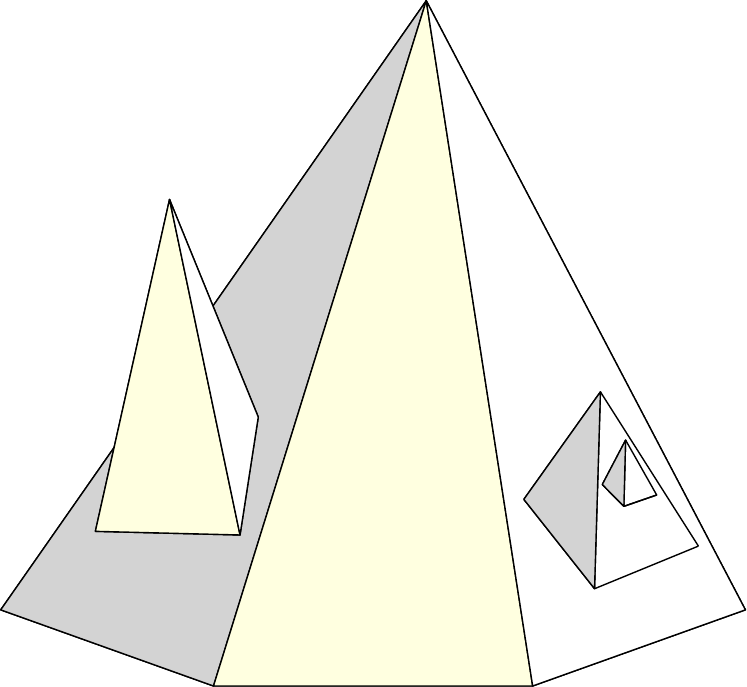}%
    \hspace{.35in}
    \includegraphics[width=.26\linewidth]{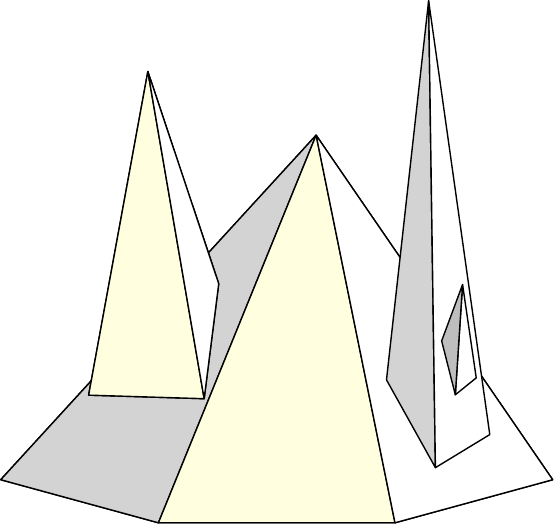}%
    \hspace{.4in}
    \includegraphics[width=.215\linewidth]{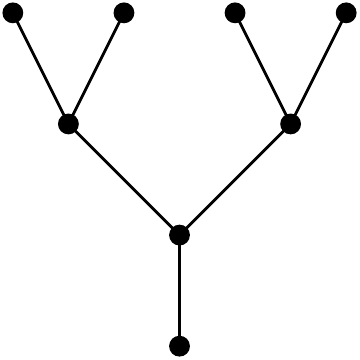}
    \caption{Two surfaces with different orderings of the maxima, but the same contour tree.}
    \label{fig:relative}
\end{figure}

\begin{figure}[h]\centering
    \includegraphics[width=.35\linewidth,height=.15\linewidth]{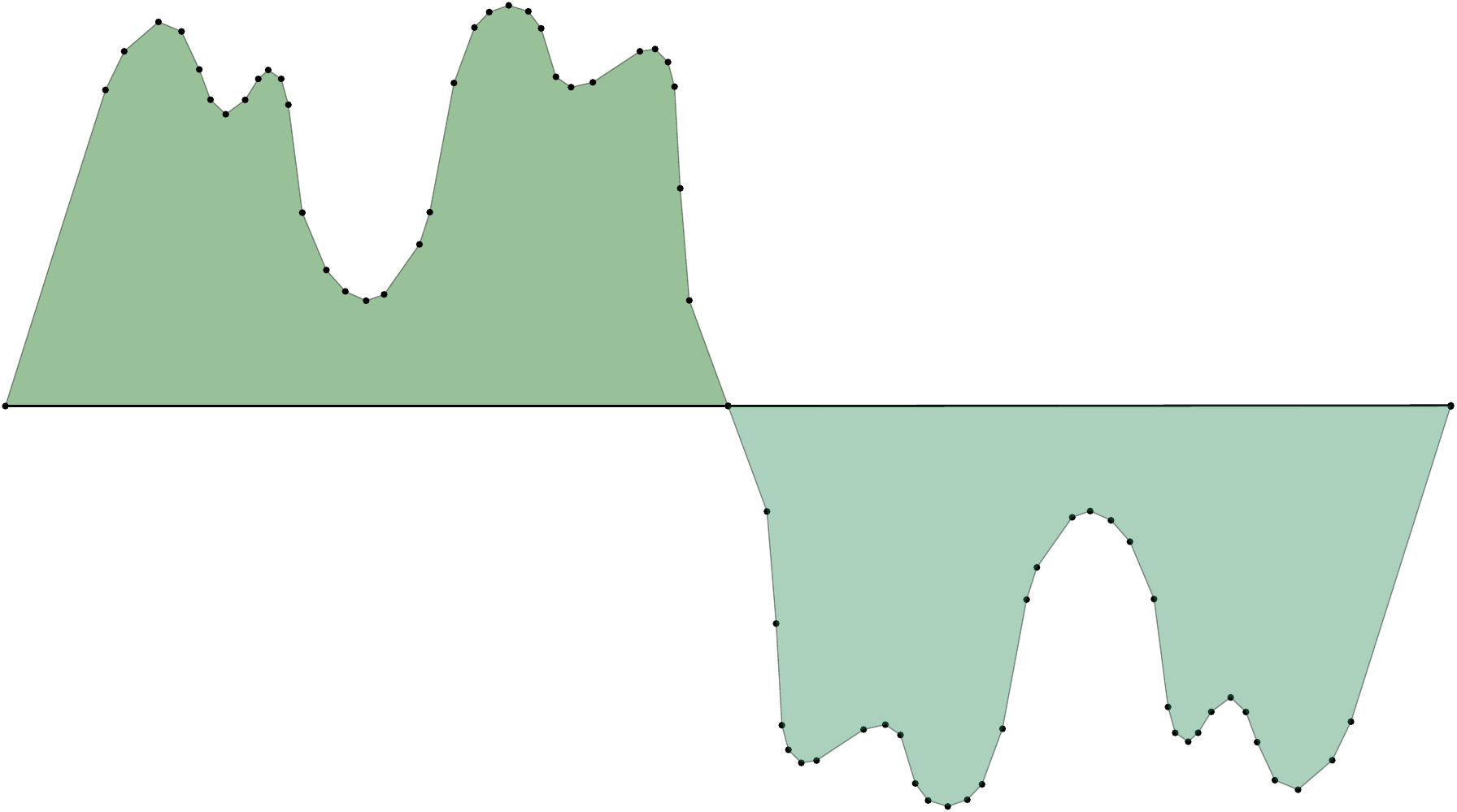}
    \hspace{.3in}
    \includegraphics[width=.35\linewidth,height=.15\linewidth]{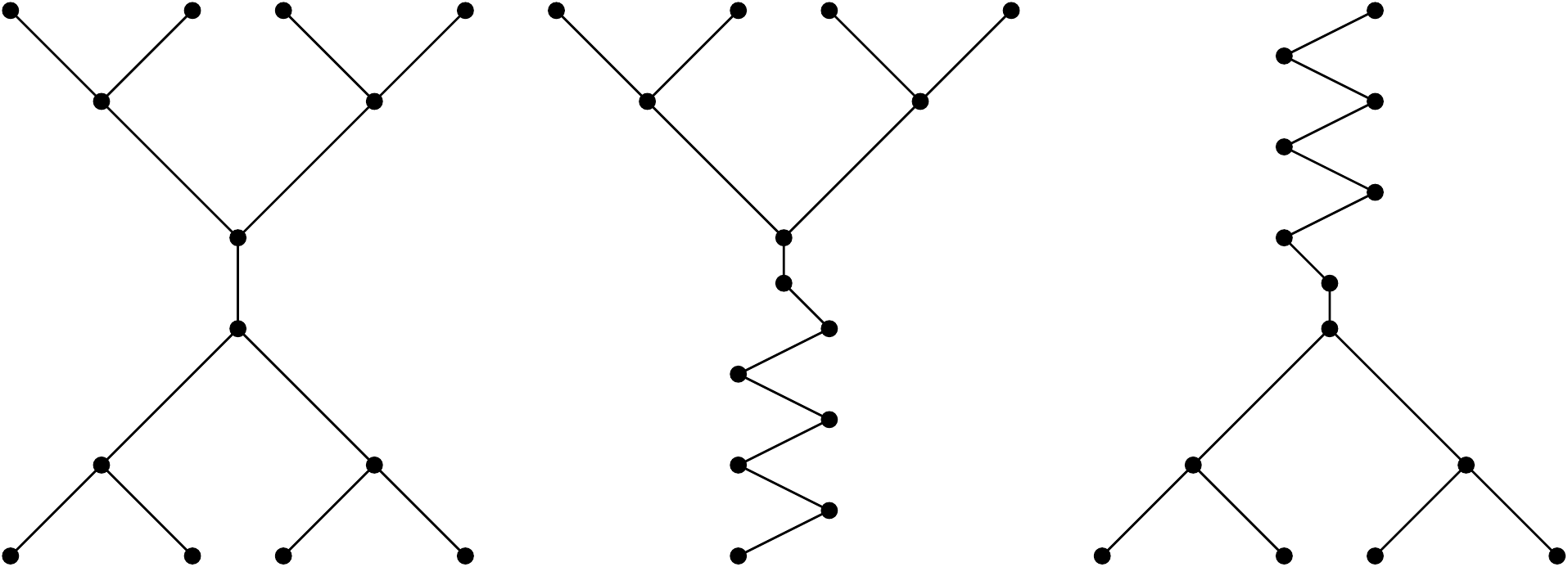}
    \caption{On left, a surface with a balanced contour tree, but whose join and split trees have long tails.  
    On right (from left to right), the contour, join and split trees.}
    \label{fig:sorting}
\end{figure}

Our main result gives an affirmative answer. Remember that we can consider
the contour tree as directed from top to bottom. For any node $v$ in the tree, let $\ell_v$ denote the length of the 
longest directed path passing through $v$. 

\begin{theorem} \label{thm:main-corr} Consider a simplicial complex $f:\MM \to \RR$, described as above, and denote
the contour tree by $T$ with vertex set (the critical points) $C(T)$. There exists an algorithm to compute the contour tree $T$
in $O(\sum_{v \in C(T)} \log \ell_v + t\alpha(t) + N)$ time. Moreover, this algorithm only
compares function values at pairs of points that are ancestor-descendant in $T$.
\end{theorem}

Essentially, the ``run time per critical point" is the height/depth of the point in the contour tree.
This bound immediately yields a run time of $O(t\log D + t\alpha(t) + N)$,
where $D$ is the diameter of the contour tree. 
This is a significant improvement for short and fat contour trees. For example, if the tree is balanced,
then we get a bound of $O(t\log\log t)$.
Even if $T$ contains a long path of length $O(t/\log t)$, but is otherwise short, we get the improved bound of $O(t\log\log t)$.

\subsection{A refined bound with optimality properties}\label{sec:more-refined}

\Thm{main-corr} is a direct corollary of a stronger but more cumbersome theorem.



\begin{definition}
\label{def:path} 
For a contour tree $T$, a \emph{leaf path} is any path in $T$ containing a leaf, 
which is also monotonic in the height values of its vertices. 
Then a \emph{path decomposition}, $P(T)$, is a partition of the vertices of $T$ into a set of vertex disjoint leaf paths.  
\end{definition}

\begin{theorem} 
\label{thm:main-alg} 
There is a deterministic algorithm to compute the contour tree, $T$, whose running time is $O(\sum_{p \in P(T)} |p|\log |p| + t\alpha(t) + N)$,
where $P(T)$ is a specific path decomposition (constructed implicitly by the algorithm).
The number of comparisons made is $O(\sum_{p \in P(T)} |p|\log |p| + N)$.
In particular, any comparisons made are only between ancestors and descendants in the contour tree.
\end{theorem}

Note that \Thm{main-corr} is a direct corollary of this statement. For any $v$, $\ell_v$ is at most
the length of the path in $P(T)$ that contains $v$.
This bound is strictly stronger,
since for any balanced contour tree, the run time bound of \Thm{main-alg} is $O(t\alpha(t) + N)$, and $O(t)$
comparisons are made.

The bound of \Thm{main-alg} may seem artificial, since it actually depends on the $P(T)$ that
is implicitly constructed by the algorithm. Nonetheless, we prove that the algorithm of \Thm{main-alg}
has strong optimality properties. For convenience, fix some value of $t$, and consider the set
of terrains ($d=2$) with $t$ critical points. The bound of \Thm{main-alg} takes values ranging
from $t$ to $t\log t$. Consider some $C \in [t,t\log t]$, and consider the set of terrains
where the algorithm makes $C$ comparisons. Then \emph{any algorithm} must make roughly $C$
comparisons in the worst-case over this set. (All further details are in \InSoCGVer{Appendix~}\Sec{lb}.)

\begin{theorem} 
\label{thm:main-lb}
There exists some absolute constant $\alpha$ such that the following holds.
For sufficiently large $t$ and any $C\in [t, t\log t]$, consider the set $\bF_C$ of terrains with $t$ critical points such that
the number of comparisons made by the algorithm of \Thm{main-alg} on these terrains is in $[C,\alpha C]$.
Any algebraic decision tree that correctly computes
the contour tree on all of $\bF_C$ has a worst case running time of $\Omega(C)$.
\end{theorem}
%
%

\subsection{Previous Work}

Contour trees were first used to study terrain maps by Boyell and Ruston, and Freeman and Morse~\cite{BoRu63,FrMo67}.
Contour trees have been applied in analysis of fluid mixing, combustion simulations,
and studying chemical systems~\cite{LaBe+06,BrWe+10,BeWe+11,BrWe+11,MaGr+11}. Carr's thesis~\cite{c-tmi-04} gives various
applications of contour trees for data visualization and is an excellent reference for contour tree definitions and algorithms.

The first formal result was an $O(N\log N)$ time algorithm for functions over 2D 
meshes and an $O(N^2)$ algorithm for higher dimensions, by van Kreveld \etal \cite{kobps-ctsssit-97}. 
Tarasov and Vyalya \cite{tv-cct-98} improved the running time to $O(N\log N)$ for the 3D case.
The influential paper of Carr \etal \cite{csa-cctad-00} improved the running time for all dimensions to $O(n\log n + N\alpha(N))$.
Pascucci and Cole-McLaughlin \cite{pc-ectls-02} provided an $O(n+t\log n)$ time algorithm for 
$3$-dimensional structured meshes. Chiang \etal \cite{cllr-sooscctmp-05} provide an unconditional $O(N+t\log t)$ algorithm.

Contour trees are a special case of Reeb graphs, a general topological representation for real-valued functions
on any manifold. Algorithms for computing Reeb graphs
is an active topic of research~\cite{sk-crgacs-91,cehnp-lrbm-03,PaScBr07,DoNa09,HaWaWe10,Pa12}, where
two results explicitly reduce to computing contour trees~\cite{TiGySi09,DoNa13}.


\section{Contour tree basics} \label{sec:basics}

We detail the basic definitions about contour trees, following the terminology of Chapter 6 of Carr's thesis \cite{c-tmi-04}.
All our assumptions and definitions
are standard for results in this area, though there is some variability in notation.
The input is a continuous piecewise-linear function $f:\MM \to \RR$, where $\MM$ is a simply connected and fully triangulated simplicial complex in $\RR^d$,
except for specially designated \emph{boundary facets}. So $f$ is explicitly defined only on the vertices of $\MM$,
and all other values are obtained by linear interpolation. 

We assume that the boundary values satisfy a special property. This is mainly for convenience
in presentation.

\begin{definition} \label{def:bound} The function $f$ is \emph{boundary critical} if the following holds.
Consider a boundary facet $F$. All vertices of $F$ have the same function value. Furthermore, all
neighbors of vertices in $F$, which are not also in $F$ itself, 
either have all function values strictly greater than or all function values strictly less than the function value at $F$.
\end{definition}

This is convenient, as we can now assume that $f$ is defined on $\RR^d$. Any point inside a boundary facet
has a well-defined height, including the infinite facet, which is required to be a boundary facet.  However, 
we allow for other boundary facets, to capture the resulting surface pieces after our algorithm makes a horizontal cut. 

We think of the dimension $d$, as constant, and assume that $\MM$ is represented in a data structure that allows constant-time access to neighboring simplices
in $\MM$ (e.g.~\cite{BoMa12}). (This is analogous to a doubly connected edge list, but for higher dimensions.)
Observe that $f:\MM \rightarrow \RR$ can be thought of as a $d$-dimensional simplicial complex living in $\RR^{d+1}$, 
where $f(x)$ is the ``height" of a point $x \in \MM$, which is encoded in the representation of $\MM$. 
Specifically, rather than writing our input as $(\MM,f)$, we abuse notation and typically just write $\MM$ to denote the lifted complex.

\begin{definition} \label{def:level} The \emph{level set} at value $h$ is the set $\{x| f(x) = h\}$.
A \emph{contour} is a connected component of a level set. An \emph{$h$-contour} is a contour where $f$-values are $h$.
\end{definition}

Note that a contour that does not contain a boundary is itself a simplicial complex of one dimension lower, and is represented (in our algorithms) as such.
We let $\delta$ and $\eps$ denote infinitesimals. Let $B_\eps(x)$ denote a ball of radius $\eps$ around $x$, and let 
$f|B_\eps(x)$ be the restriction of $f$ to $B_\eps(x)$.

\begin{definition} \label{def:deg} The \emph{Morse up-degree} of $x$ is the number of $(f(x) + \delta)$-contours of $f|B_\eps(x)$
as $\delta, \eps \rightarrow 0^+$. The \emph{Morse down-degree} is the number of $(f(x) - \delta)$-contours of $f|B_\eps(x)$
as $\delta, \eps \rightarrow 0^+$.

A \emph{regular} point has both Morse up-degree and down-degree $1$. A \emph{maximum} has Morse up-degree $0$, while
a \emph{minimum} has Morse down-degree $0$. A \emph{Morse Join} has Morse up-degree strictly greater than $1$,
while a \emph{Morse Split} has Morse down-degree strictly greater than $1$. Non-regular points are called \emph{critical}.
\end{definition}

%
%

The set of critical points is denoted by $\cV(f)$.
Because $f$ is piecewise-linear, all critical points are vertices in $\MM$. 
A value $h$ is called \emph{critical}, if $f(v) = h$, for some $v \in \cV(f)$. 
A contour is called \emph{critical}, if it contains a critical point, and it is called \emph{regular} otherwise.

The critical points are exactly where the topology of level sets change.
By assuming that our manifold is boundary critical, the vertices on a given boundary are either collectively all maxima or all minima.  
We abuse notation and refer to this entire set of vertices as a maximum or minimum.

\begin{definition} \label{def:equiv} Two regular contours $\psi$ and $\psi'$ are \emph{equivalent} 
if there exists an $f$-monotone path $p$ connecting a point in $\psi$ to $\psi'$,
such that no $x \in p$ belongs to a critical contour.
\end{definition}

This equivalence relation gives a set of \emph{contour classes}. 
Every such class maps to intervals
of the form $(f(x_i),f(x_j))$, where $x_i, x_j$ are critical points. Such a class is said
to be created at $x_i$ and destroyed at $x_j$. 

\begin{definition} \label{def:tree} The \emph{contour tree} is the graph on vertex set $\cV= \cV(f)$, where
edges are formed as follows. For every contour class that is created at $v_i$ and destroyed $v_j$,
there is an edge $(v_i,v_j)$. (Conventionally, edges are directed from higher to lower function value.)
\end{definition}

We denote the contour tree of $\MM$
by $\reeb(\MM)$. The corresponding node and edge sets are denoted as $\cV(\cdot)$ and $\cE(\cdot)$.
It is not immediately obvious that this graph is a tree, but alternate definitions of the contour tree
in~\cite{csa-cctad-00} imply this is a tree. 
Since this tree has height values associated with the vertices, we can talk about up-degrees and down-degrees in $\reeb(\MM)$.
Similar to \cite{kobps-ctsssit-97} (among others), multi-saddles are treated as a set of ordinary saddles, which can be realized 
via vertex unfolding (which can increase surface complexity if multi-saddle degrees are allowed to be super-constant).
Therefore, to simplify the presentation, for the remainder of the paper up and down-degrees are at most $2$, and total degree is at most $3$.
%


\InSoCGVer{See \Sec{techmarks} for further technical remarks on the above definitions.}
\newcommand{\technicalRemarks}{
Note that if one intersects $\MM$ with a given ball $B$, then a single contour in $\MM$ might be split into more than one contour in the intersection. 
In particular, two $(f(x)+\delta)$-contours of $f|_{B_\eps(x)}$, given by \Def{deg}, might actually be the same contour in $\MM$. 
Alternatively, one can define the up-degree (as opposed to \emph{Morse} up-degree) as the number of $(f(x)+\delta)$-contours (in the full $\MM$)
that intersect $B_\eps(x)$, a potentially smaller number. This up-degree
is exactly the up-degree of $x$ in $\reeb(\MM)$. (Analogously, for down-degree.)
When the Morse up-degree is $2$ but the up-degree is $1$, the topology of the level set changes but not by the number of connected components changing.
For example, when $d=3$ this is equivalent to the contour gaining a handle.
When $d=2$, this distinction is not necessary, since any point with Morse degree strictly greater than $1$ will have degree strictly greater than $1$ in $\reeb(\MM)$.

As Carr points out in Chapter 6 of his thesis, the term contour tree can be used for a family of related structures.
Every vertex in $\MM$ is associated with an edge in $\reeb(\MM)$, and sometimes the
vertex is explicitly placed in $\reeb(\MM)$ (by subdividing the respective edge). This is referred
to as augmenting the contour tree, and it is common to augment $\reeb(\MM)$ with all vertices.
Alternatively, one can smooth out all vertices of up-degree and down-degree $1$ to get the
unaugmented contour tree. (For $d=2$, there are no such vertices in $\reeb(\MM)$.) 
The contour tree of \Def{tree} is the typical definition in all results on output-sensitive
contour trees, and is the smallest tree that contains all the topological changes
of level sets. \Thm{main-alg} is applicable for any augmentation of $\reeb(\MM)$ with a predefined
set of vertices, though we will not delve into these aspects in this paper.
}
\InNotSoCGVer{
\subsection{Some technical remarks}
\technicalRemarks
}


\newcommand{\newintuitionSection}{
\section{A tour of the new contour tree algorithm} \label{sec:approach}
We provide a short, high-level description of the main ideas of subsequent sections. 
A longer more detailed version of this high-level description can be found in the full version \CiteFullVer.
 
\myparagraph{Cutting $\MM$ into extremum dominant pieces:} We define a simplicial complex
endowed with a height to be \emph{extremum dominant} if there exists only a single minimum, or only a single maximum.
(When formally defined in \Sec{rain}, extremum dominant complexes will allow additional trivial minima or maxima which are a small complication resulting from the following cutting procedure.)
We first cut $\MM$ into disjoint extremum dominant pieces in linear time.
Take an arbitrary maximum $x$ and imagine torrential rain at the maximum. The water flows down, wetting any point that has a non-ascending path 
from $x$. We end up with two portions, the wet part of $\MM$ and the dry part. This is similar
to \emph{watershed} algorithms used for image segmentation~\cite{RoMe00}. The wet part is obviously connected and
can be shown to be extremum dominant. We can cut along the interface of the wet and dry, remove
the wet part, and recurse on the dry parts. This is carefully done to ensure that the total complexity
of the pieces is linear.
%
%
We prove a simple contour surgery theorem that builds the contour tree of $\MM$
from the contour trees of the various pieces created.
 
Using ideas from~\cite{csa-cctad-00}, we prove that the contour tree of an extremum dominant complex is (basically) just the \emph{join tree}, 
which unlike the contour tree tracks superlevel sets rather than level sets.
Consider sweeping down the hyperplane $x_{d+1} = h$ and taking the connected components of the portion of $\MM$ above height $h$. For a terrain, these superlevel
sets are a collection of ``mounds". As we sweep downwards, these mounds keep joining each other,
until finally, we end up with all of $\MM$. The join tree tracks exactly these events. 

\myparagraph{Join trees from painted mountaintops:} Our main result is a faster algorithm for join trees.
The key idea is \emph{paint spilling}. Start with each maximum having a large can of paint, with distinct
colors for each maximum. In arbitrary order, we spill paint from each maximum, wait till it flows down,
then spill from the next, etc. Paint is viscous, and only flows down edges. Furthermore, our paints do not mix, so each edge receives a unique color, decided by the first
paint to reach it.  

Our algorithm incrementally builds the join tree from the leaves (maxima) to the root. 
Refer to the left part of \Fig{colors}.
Consider two sibling leaves $\ell_1, \ell_2$ and their common parent $v$. The leaves are maxima,
and $v$ is a join that ``merges" $\ell_1, \ell_2$. 
In that case, there are ``mounds" corresponding to $\ell_1$ and $\ell_2$ that merge
at a valley $v$. Suppose this was the entire input, and $\ell_1$ was colored blue and $\ell_2$ was colored red. 
Both mounds are colored completely blue or red, while $v$ is touched by both colors.
So this indicates that $v$ joins the blue maximum and red maximum in the join tree and is a critical point.
To process $v$, we ``merge" the colors red and blue into a new color, purple. 
In terms of the join tree, this
is equivalent to removing leaves $\ell_1$ and $\ell_2$, and making $v$ a new leaf.

Of course, things are more complicated when there are other mounds. There may be a yellow mound,
corresponding to $\ell_3$ that joins with the blue mound higher up at some vertex $u$ (see the right part of \Fig{colors}). In the join tree, $\ell_1$ and $\ell_3$
are sibling leaves, and $\ell_2$ is a sibling of some ancestor of these leaves. So we cannot
merge red and blue, until yellow and blue merge. 
A critical insight is the use of \emph{binomial heaps}~\cite{Vu78} to handle all the priorities and the merging.
This ensures that all comparisons are only made between points that are ancestor-descendant in the join tree.
There are numerous details to handle to prove the final run time bound, but this should provide the reader
some intuition behind our algorithm.

\begin{figure}[h]\centering
    \includegraphics[width=.25\linewidth]{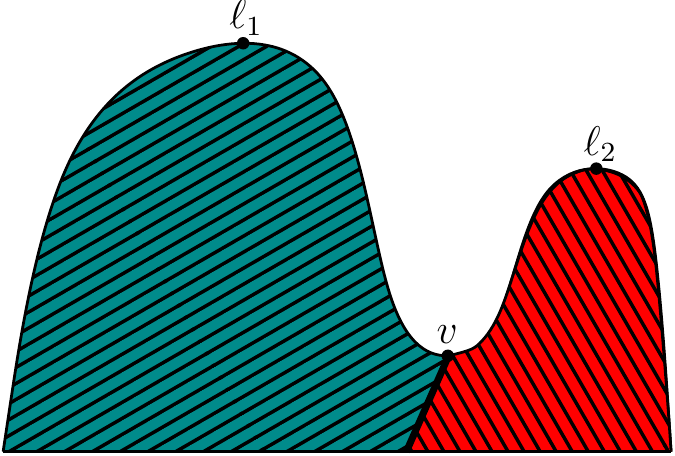}%
    \hspace{.07in}
    \includegraphics[width=.25\linewidth]{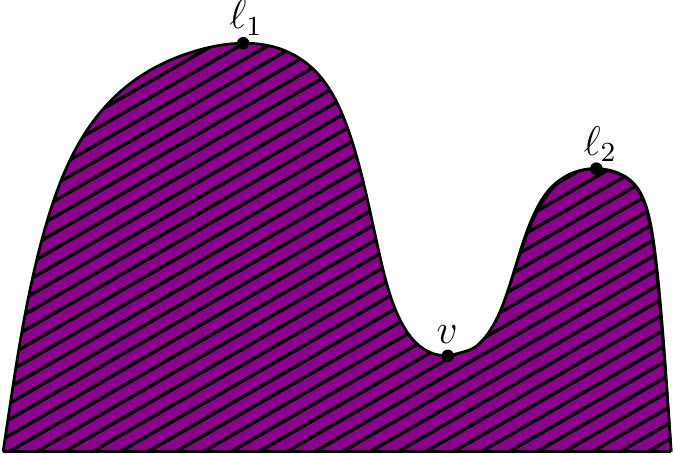}%
    \hspace{.07in}
    \includegraphics[width=.03\linewidth]{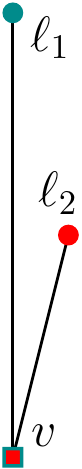}
    \hspace{.09in}
    \includegraphics[width=.3\linewidth]{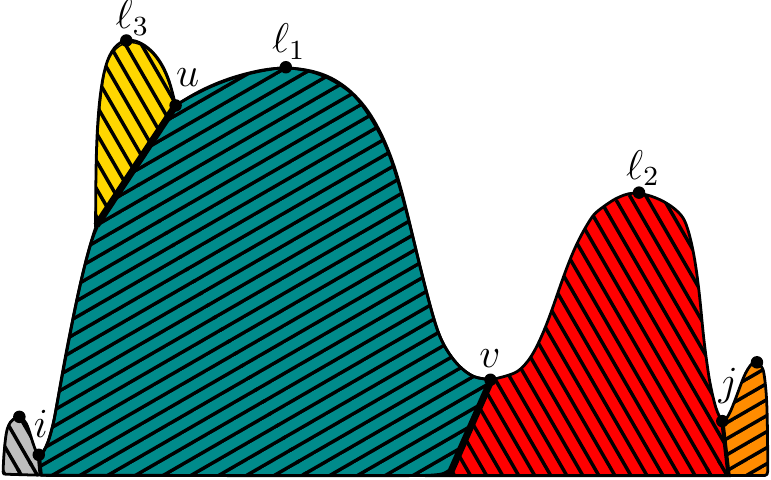}%
    \hspace{.07in}
    \includegraphics[width=.065\linewidth]{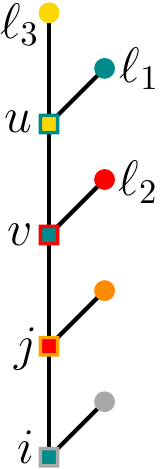}
    \caption{On the left, red and blue merge to make purple, followed by the contour tree with initial colors.  On the right, additional maxima and the resulting contour tree.}
    \label{fig:colors}
\end{figure}

}

\newcommand{\intuitionSection}{
\section{A tour of the new contour tree algorithm} \label{sec:approach}
\InNotSoCGVer{
Our final algorithm is quite technical and has numerous moving parts. However, for the $d=2$ case, 
where the input is just a triangulated terrain, the main ideas of the parts of the algorithm 
can be explained clearly.  Therefore, here we first provide a high level view of the entire result.
}
\InSoCGVer{
For the $d=2$ case, where the input is just a triangulated terrain, the main ideas of the parts of the algorithm 
can be explained clearly.  Therefore, here we provide a high level view of the entire result.
}

In the interest of presentation, the definitions and theorem statements in this section will
slightly differ from those in the main body. They may also differ from the original definitions proposed
in earlier work.

\myparagraph{Do not globally sort:} The starting point for this work is \Fig{relative}. We have two terrains
with exactly the same contour tree, but different orderings of (heights of) the critical points.
Turning it around, we cannot deduce the full height ordering of critical points from the contour tree.
Sorting all critical points is computationally unnecessary for constructing the contour tree.
In \Fig{sorting}, the contour tree consists of two balanced binary trees, one of the joins,
another of the splits. Again, it is not necessary to know the relative ordering between the mounds
on the left (or among the depressions on the right) to compute the contour tree. Yet some ordering
information is necessary: on the left, the little valleys are higher than the big central valley,
and this is reflected in the contour tree. Leaf paths in the contour tree have points
in sorted order, but incomparable points in the tree are unconstrained.
How do we sort exactly what is required, without knowing the contour tree in advance?

\subsection{Breaking $\MM$ into simpler pieces} \label{sec:break}
Let us begin with the algorithm of Carr, Snoeyink, and Axen~\cite{csa-cctad-00}. The key insight is to build
two different trees, called the join and split trees, and then merge them together into the contour tree.
Consider sweeping down the hyperplane $x_{d+1} = h$ and taking the \emph{superlevel} sets. These 
are the connected components of the portion of $\MM$ above height $h$. For a terrain, the superlevel
sets are a collection of ``mounds". As we sweep downwards, these mounds keep joining each other,
until finally, we end up with all of $\MM$. The join tree tracks exactly these events. 
Formally, let $\MM^+_v$ denote the simplicial complex induced on the subset of 
vertices which are higher than $v$.

\begin{definition} 
\label{def:int-criticalJoin}
The \emph{join tree} $\cJ(\MM)$ is built on the set $\cV$ of all critical points.
The directed edge $(u,v)$ is present when $u$ is the smallest valued vertex in $\cV$ in a connected component of $\MM^+_v$
and $v$ is adjacent (in $\MM$) to a vertex in this component. 
\end{definition}

Refer to \Fig{sorting} for the join tree of a terrain. Note that nothing happens at splits, but these
are still put as vertices in the join tree. They simply form a long path. The split tree
is obtained by simply inverting this procedure, sweeping upwards and tracking sublevel sets.

A major insight of~\cite{csa-cctad-00} is an ingeniously simple linear time procedure to construct
the contour tree from the join and split trees. So the bottleneck is computing these trees. Observe
in \Fig{sorting} that the split vertices form a long path in the join tree (and vice versa). Therefore, constructing
these trees forces a global sort of the splits, an unnecessary computation for the contour tree. Unfortunately, in general (i.e.\ unlike \Fig{sorting}) 
the heights of joins and splits may be interleaved in a complex manner, and hence the final merging of~\cite{csa-cctad-00}
to get the contour tree requires having the split vertices in the join tree. Without
this, it is not clear how to get a consistent view of both joins and splits, required for the contour tree.

Our aim is to break $\MM$ into smaller pieces, where this unnecessary computation can be avoided.

\myparagraph{Contour surgery:} We first need a divide-and-conquer lemma. Any contour $\phi$ can be associated
with an edge $e$ of the contour tree. Suppose we ``cut" $\MM$ along this contour. We prove that $\MM$
is split into two disconnected pieces, such the contour trees of these pieces is obtained
by simply cutting $e$ in $\reeb(\MM)$. Alternatively, the contour trees of these pieces can
be glued together to get $\reeb(\MM)$. This is not particularly surprising, and is fairly easy to prove
with the right definitions. The idea of loop surgery has been used to reduce Reeb graphs to contour trees~\cite{TiGySi09,DoNa13}.
Nonetheless, our theorem appears to be new and works for all dimensions.

\begin{figure}[h]\centering
    \includegraphics[width=.172\linewidth]{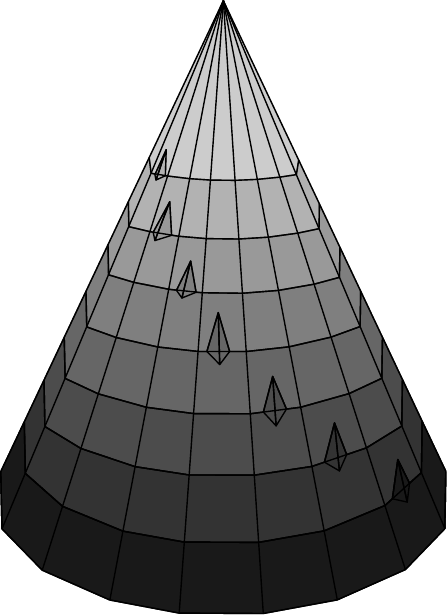}%
    \hspace{1.2in}
    \includegraphics[width=.172\linewidth]{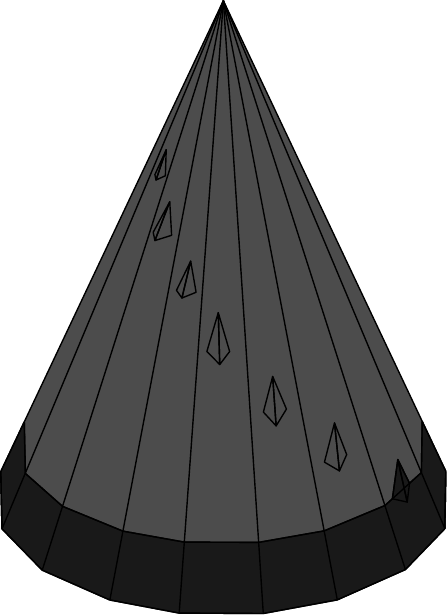}
    \caption{On left, downward rain spilling only (each shade of gray represents a piece created by each different spilling), producing a grid.  
    Note we are assuming raining was done in sorted order of the maxima (i.e.\ lowest to highest).  
    On right, flipping the direction of rain spilling.}
    \label{fig:order}
\end{figure}

\myparagraph{Cutting $\MM$ into extremum dominant pieces:} We define a simplicial complex
endowed with a height to be \emph{minimum dominant} if there exists only a single minimum.
(Our real definition is more complicated, and involves simplicial complexes
that allow additional ``trivial'' minima.)
In such a complex, there exists a non-ascending path from any point to this unique minimum.
Analogously, we can define maximum dominant complexes,
and both are collectively called extremum dominant.

We will cut $\MM$ into disjoint extremum dominant pieces, in linear time.
One way to think of our procedure is a meteorological analogy. Take an arbitrary maximum $x$,
and imagine torrential rain at the maximum. The water flows down, wetting any point that has a non-ascending path 
from $x$. We end up with two portions, the wet part of $\MM$ and the dry part. This is similar
to \emph{watershed} algorithms used for image segmentation~\cite{RoMe00}. The wet part is obviously connected,
while there may be numerous disconnected dry parts. The interface between the dry and wet parts
is a set of contours\footnote{Technically, they are not contours, but rather the limits of sequences of contours.}, 
given by the ``water line". The wet part is clearly maximum dominant, since all
wet points have a non-descending path to $x$. So we can simply cut along the interface
contours to get the wet maximum dominant piece $\MM'$. By our contour surgery theorem,
we are left with a set of disconnected dry parts, and we can recur this procedure on them.

But here's the catch. Every time we cut $\MM$ along a contour, we potentially increase
the complexity of $\MM$. Water flows in the interior of faces, and the interface will naturally
cut some faces. Each cut introduces new vertices, and a poor choice of repeated raining leads
to a large increase in complexity. Consider the left of \Fig{order}. Each raining produces a single
wet and dry piece, and each cut introduces many new vertices. If we wanted to partition this terrain into 
maximum dominant simplicial complexes, the final complexity would be forbiddingly large.

A simple trick saves the day. Unlike reality, we can choose rain to flow solely downwards or solely upwards.
Apply the procedure above to get a single wet maximum dominant $\MM'$ and a set of dry pieces.
Observe that a single dry piece $\NN$ is boundary critical with the newly introduced boundary $\phi$ (the wet-dry interface) behaving as a minimum.
So we can rain upwards from this minimum, and get a \emph{minimum dominant} portion $\NN'$.
This ensures that the new interface (after applying the procedure on $\NN$) does not
cut any face previously cut by $\phi$.
For each of the new dry pieces, the newly introduced boundary is now a maximum. So we rain downwards from there.
More formally, we alternate between raining upwards and downwards as we go down the recursion tree.
We can prove that an original face of $\MM$ is cut at most once, so the final complexity can be bounded.
In \Fig{order}, regardless of the choice of the starting maximum,  this procedure would yield (at most) two pieces, 
one maximum dominant, and one minimum dominant.

Using the contour surgery theorem previously discussed, we can build the contour tree of $\MM$
from the contour trees of the various pieces created. All in all, we prove the following theorem.

\begin{theorem} \label{thm:int-rain} There is an $O(N)$ time procedure that cuts $\MM$ into extremum dominant simplicial complexes
$\MM_1, \MM_2, \ldots$. Furthermore, given the set of contour trees $\{\cC(\MM_i)\}$, $\cC(\MM)$ can be constructed in $O(N)$ time.
\end{theorem}

\myparagraph{Extremum dominance simplifies contour trees:} We will focus on minimum dominant simplicial complexes $\MM$. By \Thm{int-rain}, it suffices
to design an algorithm for contour trees on such inputs. 
For the $d=2$ case, it helps to visualize such an input as a terrain with no minima,
except at a unique boundary face (think of a large boundary triangle that is the boundary).
All the saddles in such a terrain are necessarily joins, and there can be no splits.
In \Fig{sorting}, the portion on the left is minimum dominant in exactly this fashion, albeit in one dimension lower.
More formally, $\MM^-_v$ is connected for all $v$, so there are no splits.

We can prove that the split tree is just a path, and the contour tree is exactly
the join tree. The formal argument is a little involved, and we employ the merging procedure
of~\cite{csa-cctad-00} to get a proof. We demonstrate that the merging procedure will
actually just output the join tree, so we do not need to actually compute the split tree.
(The real definition of minimum dominant is a little more complicated,
so the contour tree is more than just the join tree. But computationally, it suffices
to construct the join tree.)

We stress the importance of this step for our approach. Given the algorithm of~\cite{csa-cctad-00},
one may think that it suffices to design faster algorithms for join trees. But this
cannot give the sort of optimality we hope for. Again, consider \Fig{sorting}. Any algorithm
to construct the true join tree must construct the path of splits, which implies
sorting all of them. It is absolutely necessary to cut $\MM$ into pieces where the cost
of building the join tree can be related to that of building $\cC(\MM)$.

\subsection{Join trees from painted mountaintops} \label{sec:join-paint}

Arguably, everything up to this point is a preamble for the main result: a faster
algorithm for join trees. Our algorithm does not require the initial input to be extremum
dominant. This is only required to relate the join trees, of the resulting subcomplexes of Theorem~\ref{thm:int-rain}, to the contour tree of the initial input $\MM$.
For clarity, we use $\NN$ to denote the input here. Note that in \Def{int-criticalJoin},
the join tree is defined purely combinatorially in terms of the 1-skeleton (the underlying graph) of $\NN$.

The join tree $\cJ(\NN)$ is a rooted tree with the dominant minimum at the root, and we direct edges downwards (towards the root). So it makes sense to talk
of comparable vs incomparable vertices.
We arrive at the main challenge: how to sort only the comparable critical points, without
constructing the join tree? The join tree algorithm of~\cite{csa-cctad-00} is a typical event-based
computational geometry algorithm. We have to step away from this viewpoint to avoid the global sort.

The key idea is \emph{paint spilling}. Start with each maximum having a large can of paint, with distinct
colors for each maximum. In arbitrary order, we spill paint from each maximum, wait till it flows down,
then spill from the next, etc. Paint is viscous, and only flows down edges. \emph{It does not
paint the interior of higher dimensional faces.} That is, this process is restricted to the 1-skeleton of $\NN$.
Furthermore, our paints do not mix, so each edge receives a unique color, decided by the first
paint to reach it.  In the following, $[n]$ denotes the set $\{1, \ldots, n\}$, for any natural number $n$.

\begin{definition} \label{def:paint1} \InSoCGVer{(Restatement of \Def{paint2}.)} Let the 1-skeleton of $\NN$ have edge set $E$ and maxima $X$.
A  \emph{painting} of $\NN$ is a map $\chi:X \cup E \to [|X|]$ with the following property. 
 Consider an edge $e$. There exists a descending path from some maximum $x$ to $e$
	consisting of edges in $E$, such that all edges along this path have the same color as $x$. 

An \emph{initial} painting has the additional property that the restriction $\chi:X \to [|X|]$ is a bijection.
\end{definition}

Note that a painting colors edges, and not vertices (except for maxima). Our definition also does not require the timing aspect
of iterating over colors, though that is one way of painting $\NN$. We begin with an initial painting, since all maximum
colors are distinct. A few comments on paint vs water. The interface between two regions of different color
is \emph{not} a contour, and so we cannot apply the divide-and-conquer
approach of contour surgery. On the other hand, painting does not cut $\NN$, so there is no increase in complexity.
Clearly, an initial painting can be constructed in $O(N)$ time. This is the tradeoff between water and paint.
Water allows for an easy divide-and-conquer, at the cost of more complexity in the input. For an extremum
dominant input, using water to divide the input $\NN$ raises the complexity too much.

Our algorithm incrementally builds $\cJ(\MM)$ from the leaves (maxima) to the root (dominant minimum). 
We say that vertex $v$ is \emph{touched} by color $c$, if there is a $c$-colored edge with lower endpoint $v$.
Let us focus on an initial painting, where the colors have 1-1 correspondence with the maxima. Refer to the left part of \Fig{colors}.
Consider two sibling leaves $\ell_1, \ell_2$ and their common parent $v$. The leaves are maxima,
and $v$ is a join that ``merges" $\ell_1, \ell_2$. 
In that case, there are ``mounds" corresponding to $\ell_1$ and $\ell_2$ that merge
at a valley $v$. Suppose this was the entire input, and $\ell_1$ was colored blue and $\ell_2$ was colored red. 
Both mounds are colored completely blue or red, while $v$ is touched by both colors.
So this indicates that $v$ joins the blue maximum and red maximum in $\cJ(\MM)$.

This is precisely how we hope to exploit the information in the painting. We prove later that
when some join $v$ has all incident edges with exactly two colors, the corresponding maxima (of those colors) are exactly
the children of $v$ in $\cJ(\MM)$. To proceed further, we ``merge" the colors red and blue into a new color, purple. In other words,
we replace all red and blue edges by purple edges. This indicates that the red and blue maxima have been
handled. Imagine flattening the red and blue mounds until reaching $v$, so that the former join $v$
is now a new maximum, from which purple paint is poured. In terms of $\cJ(\MM)$, this
is equivalent to removing leaves $\ell_1$ and $\ell_2$, and making $v$ a new leaf.
Alternatively, $\cJ(\MM)$ has been constructed up to $v$, and it remains to determine $v$'s parent.
The merging of the colors is not explicitly performed as that would be too expensive; instead we maintain a union-find
data structure for that.

Of course, things are more complicated when there are other mounds. There may be a yellow mound,
corresponding to $\ell_3$ that joins with the blue mound higher up at some vertex $u$ (see the right part of \Fig{colors}). In $\cJ(\MM)$, $\ell_1$ and $\ell_3$
are sibling leaves, and $\ell_2$ is a sibling of some ancestor of these leaves. So we cannot
merge red and blue, until yellow and blue merge. Naturally, we use priority queues to handle this issue.
We know that $u$ must also be touched by blue. So all critical vertices touched by blue
are put into a priority queue keyed by height, and vertices are handled in that order.

\begin{figure}[h]
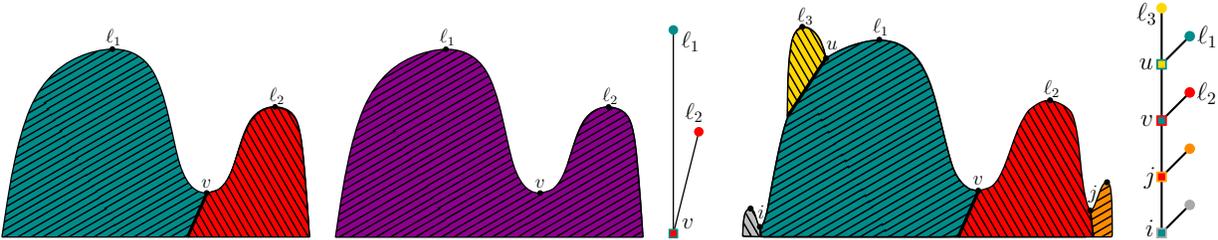
\centering
    \includegraphics[width=.25\linewidth]{figs/redBlue}%
    \hspace{.07in}
    \includegraphics[width=.25\linewidth]{figs/redBlue2}%
    \hspace{.07in}
    \includegraphics[width=.03\linewidth]{figs/redBlueTree}
    \hspace{.09in}
    \includegraphics[width=.3\linewidth]{figs/mound}%
    \hspace{.07in}
    \includegraphics[width=.065\linewidth]{figs/moundTree}
    \caption{On the left, red and blue merge to make purple, followed by the contour tree with initial colors.  On the right, additional maxima and the resulting contour tree.}
    \label{fig:colors}
\end{figure}

What happens when finally blue and red join at $v$? We merge the two colors, but now have
blue and red queues of critical vertices, which also need to be merged to get a consistent painting. 
This necessitates using a priority queue with efficient merges.
Specifically, we use \emph{binomial heaps}~\cite{Vu78}, as they provide logarithmic time merges and deletes (though other similar heaps work). 
We stress that the feasibility of the entire approach hinges on the use of such an efficient heap structure.

In this discussion, we ignored an annoying problem. Vertices
may actually be touched by numerous colors, not just one or two as assumed above. A simple solution
would be to insert vertices into heaps corresponding to all colors touching it. But there could
be super-constant numbers of copies of a vertex, and handling all these copies would lead
to extra overhead. We show that it suffices to simply put each vertex $v$ into at most two heaps, one for each
``side" of a possible join. We are guaranteed that when $v$ needs to be processed, all edges have
at most $2$ colors, because of all the color merges that previously occurred.

\myparagraph{The running time analysis:} 
 All the non-heap operations can be easily bounded by $O(t\alpha(t) + N)$ (the $t\alpha(t)$ is from the union-find data structure
 for colors). It is not hard to argue that at all times, any heap always contains a subset of a leaf to root path.
 This observation suffices to get a running time bound which is the analogue of \Thm{main-corr}, but for join trees. 
 Each heap deletion and merge can be charged to a vertex in the join tree (where each vertex gets charged only a constant number of times).
 Let $d_v$ denote the distance to the root for a vertex $v$ in the join tree, then since a heap's elements are on a single leaf to root 
 path, the size of $v$'s heap at the time an associated heap operation is made is at most $d_v$.
 Therefore, the total cost (of the heap operations) is at most $O(\sum_v \log d_v)$. This immediately proves a bound of $O(t\log D)$, where
 $D$ is the maximum distance to the root, an improvement over previous work. 
 
 However, this bound is non-optimal.  For example, for a balanced binary tree, this gives a bound of $O(t\log\log t)$, however, by using 
 an analysis involving path decompositions we can get an $O(t)$ bound.  Imagine walking from some leaf 
 towards the root. Each vertex on this path can have at most two colors (the ones getting merged), however, as we get closer to the root 
 the competition for which two colors a vertex gets assigned to grows, as the number of descendant leaf colors grows. This means that for some vertices, 
 $v$, the size of the heap for an associated heap operation will be significantly smaller than $d_v$. 
 
The intuition is that the paint spilling from the maxmima in the simplicial complex, corresponds to paint spilling from the leaves in 
the join tree, which decomposes the join tree into a set of colored paths.  Unfortunately, the situation is more complex since 
while a given color class is confined to a single leaf to root path, it may \emph{not} appear contiguously on this path, as the right part of \Fig{colors} shows. 
Specifically, in this figure the far left saddle (labeled $i$) is hit by blue paint.  However, there is another 
saddle on the far right (labeled $j$) which is not hit by blue paint.  Since this far right saddle is slightly 
higher than the far left one, it will merge into the component containing the blue mound (and also the yellow and red mounds)
before the far left one.  Hence, the vertices initially touched by blue are not contiguous in the join tree.

This non-contiguous complication along with the fact that heap size keep changing as color classes merge, causes the analysis to be technically challenging. 
We employ a variant of \emph{heavy path decompositions}, first used by Sleator and Tarjan for analyzing link/cut trees~\cite{st-dsdt-83}.
The final analysis charges expensive heap operations to long paths in the decomposition, resulting in the bound stated in \Thm{main-alg}.

\subsection{The lower bound} 

Consider a contour tree $T$ and the path decomposition $P(T)$ used to bound the running time. 
Denoting $\cost(P(T)) = \sum_{p \in P(T)} |p|\log |p|$, we construct a set of $\prod_{p \in P(T)} |p|!$
functions on a fixed domain such that each function has a distinct (labeled) contour tree. By a simple
entropy argument, any algebraic decision tree that correctly computes the contour tree on all instances
requires worst case $\Omega(\cost(P(T)))$ time. We prove that our algorithm makes $\Theta(\cost(P(T)))$ 
comparisons on all these instances. 

We have a fairly simple construction that works for terrains. In $P(T)$, consider the path
$p$ that involves the root. The base of the construction is a conical ``tent", and there
will be $|p|$ triangular faces that will each have a saddle. The heights of these saddles
can be varied arbitrarily, and that will give $|p|!$ different choices. Each of these saddles
will be connected to a recursive construction involving other paths in $P(T)$. Effectively,
one can think of tiny tents that are sticking out of each face of the main tent. The contour
trees of these tiny tents attach to a main branch of length $|p|$. Working out the details,
we get $\prod_{p \in P(T)} |p|!$ terrains each with a distinct contour tree. 
}

\InSoCGVer{

\newintuitionSection

}

\InNotSoCGVer{
\intuitionSection
}


\section{Divide and conquer through contour surgery} \label{sec:surgery}

{\bf The cutting operation:} We define a ``cut" operation on $f:\MM \rightarrow \RR$ that cuts along a regular contour to create
a new simplicial complex with an added boundary. Given a contour $\phi$, roughly speaking, this constructs the simplicial complex $\MM \setminus \phi$. 
We will always enforce the condition that $\phi$ never passes through a vertex of $\MM$.
Again, we use $\eps$ for an infinitesimally small value. We denote $\phi^+$ (resp. $\phi^-$) to be
the contour at value $f(\phi) + \eps$ (resp. $f(\phi) - \eps$), which is at distance $\eps$
from $\phi$. 

An $h$-contour is achieved by intersecting $\MM$ with the hyperplane $x_{d+1} = h$ and taking a connected component. (Think of the $d+1$-dimension
as height.) Given some point $x$ on an $h$-contour $\phi$, we can walk along $\MM$ from $x$ to determine $\phi$.
We can ``cut" along $\phi$ to get a new (possibly) disconnected simplicial complex $\MM'$. This is achieved
by splitting every face $F$ that $\phi$ intersects into an ``upper" face and ``lower" face. Algorithmically,
we cut $F$ with $\phi^+$ and take everything above $\phi^+$ in $F$ to make the upper face. Analogously, we cut with $\phi^-$ to get the lower face.
The faces are then triangulated to ensure that they are all simplices. 
\Ben{Reviewer wanted clarification on how triangulation is done. Maybe can use bottom-vertex triangulation of Clarkson 88, discussed in Matousek book?}
This creates the two new boundaries $\phi^+$ and $\phi^-$, and we maintain the property of constant $f$-value at a boundary.

Note that by assumption $\phi$ cannot cut a boundary face, and moreover all non-boundary faces
have constant size. Therefore, this process takes time linear in $|\phi|$, the number of faces $\phi$ intersects.
This new simplicial complex is denoted by $\cut(\phi,\MM)$.
We now describe a high-level approach to construct $\reeb(\MM)$ using this cutting procedure.

\Ben{Surgery should be rewritten to take instead take as input the output of cut, i.e. two complexes and a contour.  
This is how it is eventually used in Claim\ref{clm:rain-reeb}. Currently it does nothing really.}

\medskip
\fbox{
\begin{minipage}{0.9\textwidth}
{\bf $\surgery(\MM,\phi)$}

\smallskip
\begin{asparaenum}
	\item Let $\MM' = \cut(\MM,\phi)$.
	\item Construct $\reeb(\MM')$ and let $A, B$ be the nodes corresponding to the new boundaries created
	in $\MM'$. (One is a minimum and the other is maximum.)
	\item Since $A, B$ are leaves, they each have unique neighbors $A'$ and $B'$, respectively. Insert
	edge $(A',B')$ and delete $A, B$ to obtain $\reeb(\MM)$.
\end{asparaenum}
\end{minipage}}

\medskip

\InSoCGVer{The following theorems are intuitively obvious, and are proven in \Sec{proofOfSurgery}.}

\begin{theorem} \label{thm:surgery} For any regular contour $\phi$, the output of $\surgery(\MM,\phi)$ is $\reeb(\MM)$.
\end{theorem}

\newcommand{\surgeryBody}{
To prove Theorem~\ref{thm:surgery}, we require a theorem from \cite{c-tmi-04} (Theorems 6.6) that map paths in $\cC(\MM)$ to $\MM$.

\begin{theorem} \label{thm:carr-path} For every path $P$ in $\MM$, there exists a path $Q$ in the contour tree corresponding
to the contours passing through points in $P$.
For every path $Q$ in the contour tree, there exists at least one path $P$ in $\MM$ through points present
in contours involving $Q$.

In particular, for every monotone path $P$ in $\MM$, there exists a monotone path $Q$ in the contour tree to which $P$
maps, and vice versa.
\end{theorem}


Theorem~\ref{thm:surgery} is a direct consequence of the following lemma. 

\begin{lemma} \label{lem:cut} Consider a regular contour $\phi$ contained in a contour class (of an edge of $\reeb(\MM))$
$(u,v)$ and let $\MM' = \cut(\MM,\phi)$. Then $\cV(\reeb(\MM')) = \{\phi^+,\phi^-\} \cup \cV(\MM)$
and $\cE(\reeb(\MM')) = \{(u,\phi^+),(\phi^-,v)\} \cup (\cE(\MM) \setminus (u,v))$.
\end{lemma}

\begin{proof} First observe that since $\phi$ is a regular contour, the vertex set in the complex $\MM'$ is the same 
as the vertex set in $\MM$, except with the addition of the newly created vertices on $\phi^+$ and $\phi^-$.  
Moreover, $\cut(\MM,\phi)$ does not affect the local neighborhood of any vertex in $\MM$. 
Therefore since a vertex being critical is a local condition, with the exception of new boundary vertices, the 
critical vertices in $\MM$ and $\MM'$ are the same.  Finally, the new vertices on $\phi^+$ and $\phi^-$ 
collectively behave as a minimum and maximum, respectively, and so $\cV(\reeb(\MM')) = \{\phi^+,\phi^-\} \cup \cV(\MM)$.

Now consider the edge sets of the contour trees.  Any contour class in $\MM'$ (i.e.\ edge in $\cC(\MM')$) that does not involve $\phi^+$ or $\phi^-$ is also
a contour class in $\MM$. Furthermore, a maximal contour class satisfying these properties is also
maximal in $\MM$. So all edges of $\cC(\MM')$ that do not involve $\phi^+$ or $\phi^-$ are edges of $\cC(\MM)$.
Analogously, every edge of $\cC(\MM)$ not involving $\phi$ is an edge of $\cC(\MM')$.

Consider the contour class corresponding to edge $(u,v)$ of $\cC(\MM)$. There is a natural ordering
of the contours by function value, ranging from $f(u)$ to $f(v)$. All contours in this class ``above" $\phi$
form a maximal contour class in $\MM'$, represented by edge $(u,\phi^+)$. Analogously, there is another
contour class represented by edge $(\phi^-,v)$. We have now accounted for all contours in $\cC(\MM')$,
completing the proof.
\end{proof}

A useful corollary of this lemma shows that a contour actually splits the simplicial complex into
two disconnected complexes.
}
\InNotSoCGVer{
\surgeryBody
}

\begin{theorem} \label{thm:jordan} $\cut(\MM,\phi)$ consists of two
disconnected simplicial complexes.
\end{theorem}

\newcommand{\proofofDisconnected}{
\begin{proof} Denote (as in \Lem{cut}) the edge containing $\phi$ to be $(u,v)$. Suppose for contradiction that there is a path between vertices $u$ and $v$
in $\MM' = \cut(\MM,\phi)$. By \Thm{carr-path}, there is a path in $\cC(\MM')$ between $u$ and $v$. Since $\phi^+$ and $\phi^-$
are leaves in $\cC(\MM')$, this path cannot use their incident edges. Therefore by \Lem{cut},
all the edges of this path are in $\cE(\cC(\MM)) \setminus (u,v)$. So we get a cycle in $\cC(\MM)$, a contradiction.
To show that there are exactly two connected components in $\cut(\MM,\phi)$, it suffices
to see that $\cC(\MM')$ has two connected components (by \Lem{cut}) and then applying \Thm{carr-path}.
\end{proof}
}
\InNotSoCGVer{\proofofDisconnected}

%
%

\section{Raining to partition $\MM$} \label{sec:rain}

In this section, we describe a linear time procedure that partitions $\MM$ into special
\emph{extremum dominant} simplicial complexes.

\begin{definition} \label{def:dom} A simplicial complex is \emph{minimum dominant} if there exists
a minimum $x$ such that every non-minimal \emph{vertex} in the manifold has a non-ascending path to $x$.
Analogously define \emph{maximum dominant}. 
\end{definition}

The first aspect of the partitioning is ``raining''. Start at some point $x \in \MM$ and imagine rain at $x$.
The water will flow downwards along non-ascending paths and ``wet'' all the points encountered. Note that this procedure considers all points
of the manifold, not just vertices.

\begin{definition} \label{def:wet} Fix $x \in \MM$. The set of points $y \in \MM$ such that there is a non-ascending path from $x$ to $y$
is denoted by $\wet(x,\MM)$ (which in turn is represented as a simplicial complex). A point $z$ is at the \emph{interface} of $\wet(x,\MM)$ if every neighborhood of $z$
has non-trivial intersection with $\wet(x,\MM)$ (i.e.\ the intersection is neither empty nor the entire neighborhood).
\end{definition}

The following claim gives a description of the interface. 

\begin{claim} \label{clm:inter} For any $x$, each component of the interface of $\wet(x,\MM)$ contains a join vertex.
\end{claim}

\begin{proof} If $p \in \wet(x,\MM)$, all the points in any contour containing $p$ are also in $\wet(x,\MM)$.
 (Follow the non-ascending path from $x$ to $p$ and then walk along the contour.) The converse is also true,
 so $\wet(x,\MM)$ contains entire contours.

Let $\eps, \delta$ be sufficiently small as usual. Fix some $y$ at the interface.
Note that $y \in \wet(x,\MM)$. (Otherwise, $B_\eps(y)$ is dry.)
The points in $B_\eps(y)$ that lie below $y$ have a descending path from $y$ and hence must be wet.
There must also be a dry point in $B_\eps(y)$ that is above $y$, and hence,
there exists a dry, regular $(f(y)+\delta)$-contour $\phi$ intersecting $B_\eps(y)$.

Let $\Gamma_y$ be the contour containing $y$.  Suppose for contradiction
that $\forall p \in \Gamma_y$, $p$ has up-degree $1$ (see \Def{deg}). Consider the non-ascending path from $x$ to $y$ and let $z$
be the first point of $\Gamma_y$ encountered. There exists a wet, regular $(f(y) + \delta)$-contour $\psi$ 
intersecting $B_\eps(z)$. Now, walk from $z$ to $y$ along $\Gamma_y$. If all points $w$ in this walk
have up-degree $1$, then $\psi$ is the unique $(f(y)+\delta)$-contour
intersecting $B_\eps(w)$. This would imply that $\phi = \psi$, contradicting the fact that $\psi$ is wet
and $\phi$ is dry.
%
\end{proof}

Note that $\wet(x,\MM)$ (and its interface) can be computed in time linear in the size of the wet simplicial complex.
We perform a non-ascending search from $x$. Any face $F$ of $\MM$ encountered is partially (if not entirely) in $\wet(x,\MM)$.
The wet portion is determined by cutting $F$ along the interface. Since each component of the interface is a contour, this is equivalent
to locally cutting $F$ by a hyperplane. All these operations can be performed to output $\wet(x,\MM)$ in time linear in $|\wet(x,\MM)|$.

We define a simple \lift{} operation on the interface components. Consider such a component $\phi$ containing
a join vertex $y$. Take any dry increasing edge incident to $y$, and pick the point $z$ on this edge at height
$f(y) + \delta$ (where $\delta$ is an infinitesimal, but larger than the value $\eps$ used in the definition of $\cut$). 
Let $\lift(\phi)$ be the unique contour through the regular point $z$. Note that $\lift(\phi)$ is dry.
The following claim follows directly from \Thm{jordan}.

\begin{claim} \label{clm:cut-int} Let $\phi$ be a connected component of the interface. Then $\cut(\MM,\lift(\phi))$
results in two disjoint simplicial complexes, one consisting entirely of dry points.
\end{claim}

\InNotSoCGVer{
\begin{proof} By \Thm{jordan}, $\cut(\MM,\lift(\phi))$ results in two disjoint simplicial complexes. Let $\NN$ be the complex containing
the point $x$ (the argument in $\wet(x,\MM)$), and let $\NN'$ be the other complex. 
Any path from $x$ to $\NN'$ must intersect $\lift(\phi)$, which is dry. Hence $\NN'$ is dry.
%
%
\end{proof}
}

We describe the main partitioning procedure that cuts a simplicial complex $\NN$ into extremum
dominant complexes. It takes an additional input of a maximum $x$. To initialize,
we begin with $\NN$ set to $\MM$ and $x$ as an arbitrary maximum. 
When
we start, rain flows downwards. In each recursive call, the direction of rain is \emph{switched} to the 
opposite direction. This is crucial to ensure a linear running time. 
The switching is easily implemented by inverting a complex $\NN'$, achieved by negating the height values.
We can now let rain flow downwards, as it usually does in our world.

%
%
\medskip
\fbox{
\begin{minipage}{0.9\textwidth}
{\bf $\rain(x,\NN)$}

\smallskip
\begin{compactenum}
	\item Determine interface of $\wet(x,\NN)$. 
	\item If the interface is empty, simply output $\NN$. Otherwise, denote the connected components by $\phi_1, \phi_2, \ldots, \phi_k$ and set $\phi'_i = \lift(\phi_i)$.
	\item Initialize $\NN_1 = \NN$.
	\item For $i$ from $1$ to $k$: 
	\begin{compactenum}
		\item Construct $\cut(\NN_i,\phi'_i)$, consisting of dry complex $\LL_i$ and remainder $\NN_{i+1}$.
		\item Let the newly created boundary of $\LL_i$ be $B_i$. Invert $\LL_i$ so that $B_i$ is a maximum. Recursively
		call $\rain(B_i,\LL_i)$.
	\end{compactenum}
	\item Output $\NN_{k+1}$ together with any complexes output by recursive calls.
\end{compactenum}
\end{minipage}}

\medskip
For convenience, denote the total output of $\rain(x,\MM)$ by $\MM_1, \MM_2, \ldots, \MM_r$.

\begin{lemma} \label{lem:rain-1} Each output $\MM_i$ is extremum dominant.
\end{lemma}

\begin{proof} Consider a call to $\rain(x,\NN)$. If the interface is empty, then all of $\NN$
is in $\wet(x,\NN)$, so $\NN$ is trivially extremum dominant. So suppose the interface
is non-empty and consists of $\phi_1, \phi_2, \ldots, \phi_k$ (as denoted in the procedure).
By repeated applications of \Clm{cut-int}, $\NN_{k+1}$ contains $\wet(x,\MM)$. 
Consider $\wet(x,\NN_{k+1})$. The interface must exactly be $\phi_1, \phi_2, \ldots, \phi_k$.
So the only dry vertices are those in the boundaries $B_1, B_2, \ldots, B_k$. But these
boundaries are maxima.
\end{proof}

As $\rain(x,\MM)$ proceeds, new faces/simplices are created because of repeated cutting. 
The key to the running time of $\rain(x,\MM)$ is bounding the number of newly created faces, for which we have the following lemma.

\begin{lemma}\label{lem:new-verts}
A face $F \in \MM$ is cut\footnote{Technically what we are calling a single cut is done with two hyperplanes.} at most once during $\rain(x,\MM)$.
\end{lemma}
\begin{proof} Notation here follows the pseudocode of $\rain$.
First, by \Thm{jordan}, all the pieces on which $\rain$ is invoked are disjoint.
Second, all recursive calls are made on dry complexes.

Consider the first time that $F$ is cut, say, during the call to $\rain(x,\NN)$.
Specifically, say this happens when $\cut(\NN_i,\phi'_i)$ is constructed. 
$\cut(\NN_i,\phi'_i)$ will cut $F$ with two horizontal cutting planes, one $\eps$ above $\phi'_i$
and one $\eps$ below $\phi'_i$.  This breaks $F$ into lower and upper portions which are then triangulated
(there is also a discarded middle portion).  The lower portion, which is adjacent to $\phi_i$, gets included 
in $\NN_{k+1}$, the complex containing the wet points, and hence does not participate in any later recursive calls.  
The upper portion (call it $U$) is in $\LL_i$. Note that the lower
boundary of $U$ is in the boundary $B_i$. Since a recursive call is made to $\rain(B_i,\LL_i)$
(and $\LL_i$ is inverted), $U$ becomes wet. Hence $U$, and correspondingly $F$, will not be subsequently cut.
\end{proof}

The following are direct consequences of \Lem{new-verts} and the $\surgery$ procedure.

\begin{theorem} \label{thm:rain-time} The total running time of $\rain(x,\MM)$ is $O(|\MM|)$.
\end{theorem}

\InNotSoCGVer{
\begin{proof} The only non-trivial operations performed are $\wet$ and $\cut$.  
Since $\cut$ is a linear time procedure, Lemma~\ref{lem:new-verts} implies the total time for all calls to $\cut$ is $O(|\MM|)$.  
As for the $\wet$ procedure, observe that Lemma~\ref{lem:new-verts} additionally implies there are only $O(|\MM|)$ new faces created by $\rain$.
Therefore, since $\wet$ is also a linear time procedure, and no face is ever wet twice, the total time for all calls to $\wet$ is $O(|\MM|)$.
\end{proof}
}

\begin{claim} \label{clm:rain-reeb} Given $\reeb(\MM_1), \reeb(\MM_2), \ldots, \reeb(\MM_r)$, 
$\reeb(\MM)$ can be constructed in $O(|\MM|)$ time.
\end{claim}

\InNotSoCGVer{
\begin{proof} Consider the tree of recursive calls in $\rain(x,\MM)$, with each node
labeled with some $\MM_i$.
Walk through this tree in a leaf first ordering.  Each time we visit a node we connect its contour tree to the 
contour tree of its children in the tree using the $\surgery$ procedure. 
Each $\surgery$ call takes constant time, and the total time is the size of the recursion tree.
\end{proof}
}

\section{Contour trees of extremum dominant manifolds} \label{sec:extreme}

The previous section allows us to restrict attention to extremum dominant manifolds.
We will orient so that the extremum in question is always a \emph{minimum}.
We will fix such a simplicial complex $\MM$, with the dominant minimum $m^*$. 
For vertex $v$, we use $\MM^+_v$ to denote the simplicial complex obtained by only
keeping vertices $u$ such that $f(u) > f(v)$. Analogously, define $\MM^-_v$. Note that $\MM^+_v$ may contain
numerous connected components. 

\InNotSoCGVer{
The main theorem of this section asserts that contour trees of minimum dominant manifolds have a simple description.
The exact statement will require some definitions and notation.
We require the notions of \emph{join} and \emph{split} trees, as given by~\cite{csa-cctad-00}.
Conventionally, all edges are directed from higher to lower function value. 
}
\InSoCGVer{
The main theorem of this section asserts that contour trees of minimum dominant manifolds have a simple description. 
Intuitively, the cutting procedure of \Sec{rain} introduced ``stumps'' where we made cuts, which 
is why we allowed for non-dominant minima and maxima in the definition of extremum dominant manifolds.  
The punchline is that, ignoring these stumps, the contour tree of an extremum dominant manifold is equivalent to its \emph{join} (or \emph{split}) tree, defined in~\cite{csa-cctad-00}.
Hence the critical join tree below is just the join tree without these stumps.  
Additional definitions and proofs are given in \Sec{jstedm}.
}

\newcommand{\joinTreeSplitTreeMerge}{
\begin{definition} \label{def:join} The \emph{join tree} $\cJ(\MM)$ of $\MM$ is built on vertex set $\cV(\MM)$.
The directed edge $(u,v)$ is present when $u$ is the smallest valued vertex in a connected component of $\MM^+_v$
\emph{and} $v$ is adjacent to a vertex in this component (in $\MM$). The \emph{split tree} $\cS(\MM)$ is obtained
by looking at $\MM^-_v$ (or alternatively, by taking the join tree of the inversion of $\MM$). 
\end{definition}

Some basic facts about these trees. 
All outdegrees in $\cJ(\MM)$ are at most $1$, all indegree $2$ vertices are joins, all leaves are maxima,
and the global minimum is the root. All indegrees in $\cS(\MM)$ are at most $1$, all outdegree $2$
vertices are splits, all leaves are minima, and the global maximum is the root.
As these trees are rooted, we can use ancestor-descendant terminology.  
Specifically, for two adjacent vertices $u$ and $v$, $u$ is the parent of $v$ if $u$ is closer to the root 
(i.e.\ each node can have at most one parent, but can have two children).

The key observation is that $\cS(\MM)$ is trivial for a minimum dominant $\MM$.

\begin{lemma} \label{lem:split} $\cS(\MM)$ consists of:
\begin{asparaitem}
	\item A single path (in sorted order) with all vertices except non-dominant minima.
	\item Each non-dominant minimum is attached to a unique split (which is adjacent to it).
\end{asparaitem}
\end{lemma}

\begin{proof} It suffices to prove that each split $v$ has one child that is just a leaf, which
is a non-dominant minimum.  Specifically, any minimum is a leaf in $\cS(\MM)$ and thereby attached to a split, 
which implies that if we removed all non-dominant minima, we must end up with a path, as asserted above.

Consider a split $v$. For sufficiently small $\eps, \delta$, there are exactly two $(f(v) - \delta)$-contours
$\phi$ and $\psi$ intersecting $B_\eps(v)$. Both of these are regular contours. There must be a non-ascending
path from $v$ to the dominant minimum $m^*$. Consider the first edge (necessarily decreasing from $v$)
on this path. It must intersect one of the $(f(v) - \delta)$-contours, say $\phi$. By \Thm{jordan}, $\cut(\MM,\phi)$ has
two connected components, with one (call it $\LL$) having $\phi^-$ as a boundary maximum. This complex
contains $m^*$ as the non-ascending path intersects $\phi$ only once. Let the other
component be called $\MM'$.

Consider $\cut(\MM',\psi)$ with connected component $\NN$ having $\psi^-$ as a boundary. $\NN$
does not contain $m^*$, so any path from the interior of $\NN$ to $m^*$ must intersect the boundary $\psi^-$.
But the latter is a maximum in $\NN$, so there can be no non-ascending path from the interior to $m^*$.
Since $\MM$ is overall minimum dominant, the interior of $\NN$ can only contain a single vertex $w$, a non-dominant
minimum.

The split $v$ has two children in $\cS(\MM)$, one in $\NN$ and one in $\LL$. The child in $\NN$ can only
be the non-dominant minimum $w$, which is a leaf. 
\end{proof}

It is convenient to denote the non-dominant minima as $m_1, m_2, \ldots, m_k$
and the corresponding splits (as given by the lemma above) as $s_1, s_2, \ldots, s_k$. 

Using the above lemma we can now prove that computing the contour tree for a minimum dominant manifold 
amounts to computing its join tree.  Specifically, to prove our main theorem, we rely on the correctness of the
merging procedure from~\cite{csa-cctad-00} that constructs the contour tree from the join and split trees. 
It actually constructs the \emph{augmented contour tree} $\cA(\MM)$, which is 
obtained by replacing each edge in the contour tree with a path of all regular vertices 
(sorted by height) whose corresponding contour belongs to the equivalence class of that edge.

Consider a tree $T$ with a vertex $v$ of in and out degree at most $1$.
\emph{Erasing} $v$ from $T$ is the following operation: if $v$ is a leaf, just delete $v$. Otherwise, 
delete $v$ and connect its neighbors by an edge (i.e.\ smooth $v$ out). This tree is denoted by $T \ominus v$.

\medskip
\fbox{
\begin{minipage}{0.9\textwidth}
{\bf $\merge(\cJ(\MM),\cS(\MM))$}

\smallskip
\begin{compactenum}
	\item Set $\cJ = \cJ(\MM)$ and $\cS = \cS(\MM)$.
	\item Denote $v$ as a \emph{candidate} if the sum of its indegree in $\cJ$ and outdegree in $\cS$ is $1$. 
	\item Add all candidates to queue.
	\item While candidate queue is non-empty: 
	\begin{compactenum}
		\item Let $v$ be head of queue. If $v$ is leaf in $\cJ$, consider its edge in $\cJ$. Otherwise
		consider its edge in $\cS$. In either case, denote the edge by $(v,w)$.
		\item Insert $(v,w)$ in $\cA(\MM)$. 
		\item Set $\cJ = \cJ \ominus v$ and $\cS = \cS \ominus v$. Enqueue any new candidates.
	\end{compactenum}
	\item Smooth out all regular vertices in $\cA(\MM)$ to get $\cC(\MM)$.
\end{compactenum}
\end{minipage}}

\medskip
}

\InNotSoCGVer{\joinTreeSplitTreeMerge}

\begin{definition} 
\label{def:criticalJoin}
The \emph{critical join tree} $\jc(\MM)$ is built on the set $V'$ of all
critical points other than the non-dominant minima. 
The directed edge $(u,v)$ is present when $u$ is the smallest valued vertex in $V'$ in a connected component of $\MM^+_v$
and $v$ is adjacent (in $\MM$) to a vertex in this component. \InSoCGVer{The \emph{join tree}, $\cJ(\MM)$, is defined analogously, but instead on $\cV(\MM)$.}
%
%
\end{definition}

\InSoCGVer{Each non-dominant minimum, $m_i$, connects to the contour tree at a corresponding split $s_i$. We have the following (see \Sec{jstedm} for more details).}

\begin{theorem} \label{thm:contour-tree} Let $\MM$ have a dominant minimum. 
The contour tree $\cC(\MM)$ consists of all edges $\{(s_i, m_i)\}$ and $\jc(\MM)$.
\end{theorem}

\newcommand{\proofofContourJoinEquiv}{
\begin{proof} We first show that $\cA(\MM)$ is $\cJ(\MM) \ominus \{m_i\}$ with edges $\{(s_i,m_i)\}$.
We have flexibility in choosing the order of processing in $\merge$. We first put the non-dominant
maxima $m_1, \ldots, m_k$ into the queue. As these are processed, the edges $\{(s_i,m_i)\}$ are inserted 
into $\cA(\MM)$. Once all the $m_i$'s are erased, $\cS$ becomes a path, so all outdegrees are at most $1$.
The join tree is now $\cJ(\MM) \ominus \{m_i\}$. We can now process $\cJ$ leaf by leaf, and all
edges of $\cJ$ are inserted into $\cA(\MM)$.

Note that $\cC(\MM)$ is obtained by smoothing out all regular points from $\cA(\MM)$. Similarly, smoothing out regular
points from $\cJ(\MM) \ominus \{m_i\}$ yields the edges of $\jc(\MM)$.
\end{proof}
}

\InNotSoCGVer{\proofofContourJoinEquiv}

\begin{Remark}\label{rem:joinAndCriticalJoin}
The above theorem, combined with the previous sections, implies that in order to get an efficient contour tree algorithm, 
it suffices to have an efficient algorithm for computing $\jc(\MM)$.
Due to minor technicalities, it is easier to phrase the following section instead in terms of computing $\cJ(\MM)$ efficiently.  
Note however that for minimum dominant complexes output by $\rain$, converting between $\jc$ and $\cJ$ is trivial, as $\cJ$ is just $\jc$ 
with each non-dominant minimum $m_i$ augmented along the edge leaving $s_i$.
\end{Remark}

\section{Painting to compute contour trees}
\label{sec:paint}

The main algorithmic contribution is a new algorithm for computing join trees of any triangulated simplicial
complex $\MM$.

{\bf Painting:} The central tool is a notion 
of \emph{painting} $\MM$. Initially associate a color with each maximum. Imagine there being a large
can of paint of a distinct color at each maximum $x$. We will spill different paint from each maximum and watch it flow down.
This is analogous to the raining of \Sec{rain}, but paint is a much more viscous liquid.
\emph{So paint only flows down edges, and it does not color the interior of higher dimensional faces.} Furthermore, paints
do not mix, so every edge of $\MM$ gets a unique color. This process (and indeed the entire algorithm)
works purely on the 1-skeleton of $\MM$, which is just a graph.

\InNotSoCGVer{We now restate \Def{paint1}.}

\begin{definition} \label{def:paint2} Let the 1-skeleton of $\MM$ have edge set $E$ and maxima $X$.
A  \emph{painting} of $\MM$ is a map $\chi:X \cup E \to [|X|]$ with the following property. 
 Consider an edge $e$. There exists a descending path from some maximum $x$ to $e$
	consisting of edges in $E$, such that all edges along this path have the same color as $x$. 

An \emph{initial} painting has the additional property that the restriction $\chi:X \to [|X|]$ is a bijection.
\end{definition}


\begin{definition} \label{def:color-set} Fix a painting $\chi$ and vertex $v$.
\begin{asparaitem}
	 \item An \emph{up-star} of $v$ is the set of edges that all connected to a fixed component of $\MM^+_v$.
	 \item A vertex $v$ is \emph{touched by color $c$} if $v$ is incident to a $c$-colored
	 edge with $v$ at the lower endpoint. For $v$, $\col(v)$ is the set of colors that touch $v$.
	 \item A color $c \in \col(v)$ \emph{fully touches} $v$ if all edges in an up-star are colored $c$.
	 \item For any maximum $x\in X$, we say that $x$ is both touched and fully touched by $\chi(x)$.
\end{asparaitem}
\end{definition}


\subsection{The data structures} \label{sec:struct}

\noindent
{\bf The binomial heaps $\touch(c)$:} For each color $c$, $\touch(c)$ is a subset of vertices touched by $c$,
This is stored as a \emph{binomial max-heap} keyed by vertex heights. Abusing notation, $\touch(c)$ refers
both to the set and the data structure used to store it.

\medskip
\noindent
{\bf The union-find data structure on colors:} We will repeatedly perform unions
of classes of colors, and this will be maintained as a standard union-find data structure.
For any color $c$, $\rep(c)$ denotes the representative of its class. 

\medskip
\noindent
{\bf The stack $\stack$:} This consists of non-extremal critical points, with monotonically increasing
heights as we go from the base to the head.
\ignore{
Each point $x \in \stack$ has an associated subset of $\col(x)$, denoted $\mcol(x)$.
Both $\mcol(x)$ and its complement are stored as hash table. So lookups, inserts, and deletes
are in these sets are all constant time operations. The stack is guaranteed to satisfy 
the following invariants.
\begin{asparaitem}
	\item For every $x \in \stack$: For every $c \in \mcol(x)$, $x$ is the highest element
	in $T(c)$. Furthermore, $c = \rep(c)$.
	\item Consider $x, y \in \stack$ such that $y$ was pushed on $x$. There exists $c \in \col(x) \setminus
	\mcol(x)$ such that $x$ is not highest in $T(c)$ but $y$ is highest in $T(c)$.
\end{asparaitem}
}

\medskip
\noindent
{\bf Attachment vertex $\h(c)$:} For each color $c$, we maintain a critical point $\h(c)$ of this color.
We will maintain the guarantee that the portion of the contour tree above (and including) $\h(c)$ has already been constructed.

\subsection{The algorithm} \label{sec:algo}

We formally describe the algorithm below. 
We require a technical definition of \emph{ripe} vertices.

\begin{definition} \label{def:ripe} A vertex $v$ is \emph{ripe} if: for all $c \in \col(v)$, $v$
is present in $T(\rep(c))$ and is also the highest vertex in this heap. 
\end{definition} 

\medskip
\fbox{
\begin{minipage}{0.9\textwidth}
{\bf $\init(\MM)$}

\smallskip
\begin{compactenum}
	\item Construct an initial painting of $\MM$ using a descending BFS from maxima that does not explore previously colored edges.
	\item Determine all critical points in $\MM$. For each $v$, look at $(f(v) \pm \delta)$-contours in $f|_{B_\eps(v)}$ 
	to determine the up and down degrees. 
	\item Mark each critical $v$ as unprocessed.
	\item For each critical $v$ and each up-star, pick an arbitrary color $c$ touching $v$. Insert $v$ into $T(c)$.
	\item Initialize $\rep(c) = c$ and set $\h(c)$ to be the unique maximum colored $c$.
	\item Initialize $K$ to be an empty stack.
\end{compactenum}
\end{minipage}}

\medskip
\fbox{
\begin{minipage}{0.9\textwidth}
{\bf $\build(\MM)$}

\smallskip
\begin{compactenum}
	\item Run $\init(\MM)$.
	\item While there are unprocessed critical points:
	\begin{compactenum}
		\item Run $\update(\stack)$. Pop $\stack$ to get $h$.
		\item Let $\cur(h) = \{\rep(c) | c \in \col(h)\}$.
		\item For all $c' \in \cur(h)$:
		\begin{compactenum}
			\item Add edge $(\h(c'),h)$ to $\cJ(\MM)$.
			\item Delete $h$ from $T(c')$.
		\end{compactenum}
	\item Merge heaps $\{T(c') | c' \in \cur(h)\}$.
	\item Take union of $\cur(h)$ and denote resulting color by $\widehat{c}$.
	\item Set $\h(\widehat{c}) = h$ and mark $h$ as processed.
	\end{compactenum}
\end{compactenum}
\end{minipage}}

\medskip
\fbox{
\begin{minipage}{0.9\textwidth}
{\bf $\update(\stack)$}

\smallskip
\begin{compactenum}
	\item If $K$ is empty, push arbitrary unprocessed critical point $v$.
	\item Let $h$ be the head of $\stack$.
	\item While $h$ is not ripe:
	\begin{compactenum}
		\item Find $c \in \col(h)$ such that $h$ is not the highest in $T(\rep(c))$.
		\item Push the highest of $T(\rep(c))$ onto $\stack$, and update head $h$.
	\end{compactenum}
\end{compactenum}
\end{minipage}}

\bigskip

A few simple facts:
\begin{compactitem}
	\item At all times, the colors form a valid painting.
	\item Each vertex is present in at most $2$ heaps. After processing, it is removed from all heaps.
	\item After $v$ is processed, all edges incident to $v$ have the same (representative) color. 
	\item Vertices on the stack are in increasing height order. 
\end{compactitem}

\begin{observation}
\label{obs:twoQueues}
 Each unprocessed vertex is always in exactly one queue of the colors in each of its up-stars.  Specifically, 
 for a given up-star of a vertex $v$, $\init(\MM)$ puts $v$ into the queue of exactly one of the colors of the up-star, say $c$.  As time goes on this queue 
 may merge with other queues, but while $v$ remains unprocessed, it is only ever (and always) in the queue of $\rep(c)$,  
 since $v$ is never added to a new queue and is not removed until it is processed.  
 In particular, finding the queues of a vertex in $\update(K)$ requires at most two union find operations (assuming each vertex records its two colors from $\init(\MM)$).
\end{observation}

\medskip

\subsection{Proving correctness} \label{sec:correct}

\ignore{
\begin{claim} \label{clm:process} Assume the invariant. Any vertex with a non-increasing path to some $\h(c)$
for $c \in \pal$ has been processed.
\end{claim}

\begin{proof} This is a direct consequence of \Thm{carr-mono}, which relates monotone paths in $\MM$ to $\cC(\MM)$.
Since the subtree of $\cJ(\MM)$ (which is basically the subtree of $\cC(\MM)$) rooted at $\h(c)$ 
has been found, all vertices in this subtree must be processed. These are all the vertices with non-increasing paths
to $\h(c)$ in $\cC(\MM)$, which by \Thm{carr-mono} is the same as those in $\cC(\MM)$.
\end{proof}
}

Our main workhorse is the following technical lemma. In the following, the current color of an edge, $e$, is the value of $\rep(\chi(e))$, 
where $\chi(e)$ is the color of $e$ from the initial painting.


\begin{lemma} \label{lem:full} Suppose vertex $v$ is connected to a component $\PP$ of $\MM^+_v$
by an edge $e$ which is currently colored $c$. Either all edges in $\PP$ are currently colored $c$, or there
exists a critical vertex $w \in \PP$ fully touched by $c$ and touched by another color.
\end{lemma}

\begin{proof} Since $e$ has color $c$,
there must exist vertices in $\PP$ touched by $c$. Consider the highest
vertex $w$ in $\PP$ that is touched by $c$ and some other color. If no such vertex exists,
this means all edges incident to a vertex touched by $c$ are colored $c$. By walking through
$\PP$, we deduce that all edges are colored $c$. 

So assume $w$ exists. Take the $(f(w)+\delta)$-contour $\phi$ that intersects $B_\eps(w)$
and intersects some $c$-colored edge incident to $w$. Note that all edges intersecting $\phi$ are also colored $c$,
since $w$ is the highest vertex to be touched by $c$ and some other color. (Take the path of $c$-colored
edges from the maximum to $w$. For any point on this path, the contour passing through this point must
be colored $c$.) Hence, $c$ fully touches $w$. 
But $w$ is touched by another color, and the corresponding edge cannot intersect $\phi$. So $w$
must have up-degree $2$ and is critical.
\end{proof}

\begin{corollary}
\label{cor:terminate}
 Each time $\update(K)$ is called, it terminates with a ripe vertex on top of the stack.
\end{corollary}
\begin{proof}
 $\update(K)$ is only called if there are unprocessed vertices remaining, and so by the time we reach step 3 in 
 $\update(K)$, the stack has some unprocessed vertex $h$ on it.  
 If $h$ is ripe, then we are done, so suppose otherwise.
 
 Let $\PP$ be one of the components of $\MM^+_h$.  By construction, $h$ was put in the heap of some initial 
 adjacent color $c$.  Therefore, $h$ must be in the current heap of $\rep(c)$ (see \Obs{twoQueues}).  
 Now by \Lem{full}, either all edges in $\PP$ are colored $\rep(c)$ or  
 there is some vertex $w$ fully touched by $\rep(c)$ and some other color.  
 The former case implies that if there are any unprocessed vertices in $\PP$ then they are all in $T(\rep(c))$, 
 implying that $h$ is not the highest vertex and a new higher up unprocessed vertex will be put on the stack for the next iteration of the while loop.  
 Otherwise, all the vertices in $\PP$ have been processed.  
 However, it cannot be the case that all vertices in all components of $\MM^+_h$ have already been processed, 
 since this would imply that $h$ was ripe, and so one can apply the same argument to the other non-fully processed component. 
 
 Now consider the latter case, where we have a non-monochromatic vertex $w$.
 In this case $w$ cannot have been processed (since after being processed it is touched only by one color), 
 and so it must be in $T(\rep(c))$ since it must be in some heap of a color in each up-star (and one up-star is entirely colored $\rep(c)$).
 As $w$ lies above $h$ in $\MM$, this implies $h$ is not on the top of this heap.
\end{proof}


\begin{claim} \label{clm:upstar} Consider a ripe vertex $v$ and take the up-star connecting
to some component of $\MM^+_v$. All edges in this component and the up-star have the same color.
\end{claim}

\begin{proof} Let $c$ be the color of some edge in this up-star.
By ripeness, $v$ is the highest in $T(\rep(c))$.
Denote the component of $\MM^+_v$ by $\PP$.
By \Lem{full}, either all edges in $\PP$ are colored $\rep(c)$ or there exists critical vertex $w \in \PP$
fully touched by $\rep(c)$ and another color. In the latter case, $w$ has not been processed,
so $w \in T(\rep(c))$ (contradiction to ripeness). Therefore, all edges in $\PP$ are colored $\rep(c)$.
\end{proof}

\begin{claim} \label{clm:process} The partial output on the processed vertices is exactly
the restriction of $\cJ(\MM)$ to these vertices.
\end{claim}

\begin{proof} More generally, we prove the following: all outputs on processed vertices
are edges of $\cJ(\MM)$ and for any current color $c$, $\h(c)$ is the lowest processed vertex
of that color. We prove this by induction on the processing order. The base case is trivially
true, as initially the processed vertices and attachments of the color classes are the set of maxima.
For the induction step, consider the situation
when $v$ is being processed.

Since $v$ is being processed, we know by \Cor{terminate} that it is ripe. Take any up-star of $v$, and the corresponding component $\PP$
of $\MM^+_v$ that it connects to. By \Clm{upstar}, all edges in $\PP$ and the up-star
have the same color (say $c$). If some critical vertex in $\PP$ is not processed,
it must be in $T(c)$, which violates the ripeness of $v$.
Thus, all critical vertices in $\PP$ have been processed, and so by the induction hypothesis, the restriction of $\cJ(\MM)$ to $\PP$ has been correctly computed.
Additionally, since all critical vertices in $\PP$ have processed, they all have the same color $c$ of the lowest critical vertex in $\PP$.
Thus by the strengthened induction hypothesis, this lowest critical vertex is $\h(c)$.

If there is another component of $\MM^+_v$, the same argument implies 
the lowest critical vertex in this component is $\h(c')$ (where $c'$ is the color of edges in the respective component).
Now by the definition of $\cJ(\MM)$, the critical vertex $v$ connects to the lowest critical vertex in each component of $\MM^+_v$, 
and so by the above $v$ should connect to $\h(c)$ and $\h(c')$, which is precisely what $v$ is connected to by $\build(\MM)$.
Moreover, $\build$ merges the colors $c$ and $c'$ and correctly sets $v$ to be the attachment, 
as $v$ is the lowest processed vertex of this merged color (as by induction $\h(c)$ and $\h(c')$ were the lowest vertices before merging colors).

%
\end{proof}

\begin{theorem}
\label{thm:correct}
 Given an input complex $\MM$, $\build(\MM)$ terminates and outputs $\cJ(\MM)$.
\end{theorem}
\begin{proof}
 First observe that each vertex can be processed at most once by $\build(\MM)$.  By \Cor{terminate}, we know that as long as there 
 is an unprocessed vertex, $\update(K)$ will be called and will terminate with a ripe vertex which is ready to be processed. 
 Therefore, eventually all vertices will be processed, and so by \Clm{process} the algorithm will terminate having 
 computed $\cJ(\MM)$.
\end{proof}

\InNotSoCGVer{
\subsection{Running Time}
\label{sec:runTime}

We now bound the running time of the algorithm of \Sec{algo}.  
In subsequent sections, through a sophisticated charging argument, this bound is then related to matching upper 
and lower bounds in terms of path decompositions. 
Therefore, it will be useful to set up some terminology that can 
be used consistently in both places.  Specifically, the path decomposition bounds will be purely combinatorial 
statements on colored rooted trees, and so the terminology is of this form.

Any tree $T$ considered in following will be a rooted binary tree\footnote{Note 
that technically the trees considered should have a leaf vertex hanging below the root  
in order to represent the global minimum of the complex.  This vertex is 
(safely) ignored to simplify presentation.} where the height of a vertex is its distance 
from the root $r$ (i.e.\ conceptually $T$ will be a join tree with $r$ at the bottom).  
As such, the children of a vertex $v\in T$ are the adjacent vertices of larger height
(and $v$ is the parent of such vertices).  Then the subtree rooted at $v$, denoted $T_v$ consists of the graph induced on all 
vertices which are descendants of $v$ (including $v$ itself).  For two vertices $v$ and $w$ in $T$ let $d(v,w)$ denote the 
length of the path between $v$ and $w$.
We use $A(v)$ to denote the set of ancestors of $v$.
For a set of nodes $U$, $A(U) = \bigcup_{u \in U} A(u)$.

\begin{definition}
\label{def:leafAssign}
 A \emph{leaf assignment} $\chi$ of a tree $T$ assigns \emph{two} distinct leaves to each internal vertex $v$,
 one from the left child and one from the right child subtree of $v$ (naturally if $v$ has only one child it is assigned 
 only one color).
\end{definition}

For a vertex $v\in T$, we use $H_v$ to denote the \emph{heap} at $v$.
Formally, $H_v = \{u | u \in A(v), \chi(u) \cap L(T_v) \neq \emptyset\}$, where $L(T_v)$ is the set of leaves of $T_v$.
In words, $H_v$ is the set of ancestors of $v$ which are colored by some leaf in $T_v$.


\begin{definition}
\label{def:initialColoring}
Note that the subroutine $\init(\MM)$ from \Sec{algo} naturally defines a leaf assignment to $\cJ(\MM)$ 
according to the priority queue for each up-star we put a given vertex in.  Call this the \emph{initial coloring}
of the vertices in $\cJ(\MM)$.  Note also that this initial coloring defines the $H_v$ values for all $v\in \cJ(\MM)$.
\end{definition}

The following lemma should justify these technical definitions.

\begin{lemma}
\label{lem:runTimeUpper}
Let $\MM$ be a simplicial complex with $t$ critical points.  For every vertex in $\cJ(\MM)$, 
let $H_v$ be defined by the initial coloring of $\MM$.
The running time of $\build(\MM)$ is $O(N+t\alpha(t) + \sum_{v \in \cJ(\MM)} \log |H_v|)$.
\end{lemma}
}
\newcommand{\proofofRunTime}{
\begin{proof}
First we look at the initialization procedure $\init(\MM)$.  This procedure runs in $O(N)$ time.
Indeed, the painting procedure consists of several BFS's but as each vertex is only explored by one of the BFS's, it is linear time overall.
Determining the critical points is a local computation on the neighborhood of each vertex as so is linear (i.e.\ each edge is viewed at most twice).
Finally, each vertex is inserted into at most two heaps and so initializing the heaps takes linear time in the number of vertices.

Now consider the union-find operations performed by $\build$ and $\update$.  
Initially the union find data structure has a singleton component for each leaf (and no new components are ever created), 
and so each union-find operation takes $O(\alpha(t))$ time.
For $\update$, by \Obs{twoQueues}, each iteration of the while loop requires a constant number of finds (and no unions).
Specifically, if a vertex is found to be ripe (and hence processed next) then these can be charged to that vertex.
If a vertex is not ripe, then these can be charged to the vertex put on the stack.
As each vertex is put on the stack or processed at most once, $\update$ performs $O(t)$ finds overall.
Finally, $\build(\MM)$ performs one union and at most two finds for each vertex.  Therefore the total number of union find 
operations is $O(t)$.

For the remaining operations, observe that for every iteration of the loop in $\update$, a vertex is pushed onto the stack and each
vertex can only be pushed onto the stack once (since the only way it leaves the stack is by being processed). 
Therefore the total running time due to $\update$ is linear (ignoring the find operations).

What remains is the time it takes to process a vertex $v$ in $\build(\MM)$.  
In order to process a vertex there are a few constant time operations, union-find operations, and queue operations.
Therefore the only thing left to bound are the queue operations.
Let $v$ be a vertex in $\cJ(\MM)$, and let $c_1$ and $c_2$ be its children (the same argument holds if $v$ has only one child).
At the time $v$ is processed, the colors and queues of all vertices in a given component of $\MM^+_v$ have merged together.
In particular, when $v$ is processed we know it is ripe and so all vertices above $v$ in each component of $\MM^+_v$ have been processed, implying 
these merged queues are the queues of the current colors of $c_1$ and $c_2$.  Again since $v$ is ripe, it must be on the top of these queues and so 
the only vertices left in these queues are those in $H_{c_1}$ and $H_{c_2}$. 

Now when $v$ is handled, three queue operations are performed.  Specifically, $v$ is removed from the queues of $c_1$ and $c_2$, and then the queues are 
are merged together.  By the above arguments the sizes of the queues for each of these operations are $H_{c_1}$, $H_{c_2}$, and $H_v$, respectively.  
As merging and deleting takes logarithmic time in the heap size for binomial heaps, the claim now follows.
\end{proof}
}
\InNotSoCGVer{
\proofofRunTime
}
\InSoCGVer{
\subsection{Upper Bounds for Running Time}\label{sec:runTime}

The algorithm $\build(\MM)$ processes vertices in $\cJ(\MM)$ one at a time.  The main processing cost  
comes from priority queue operations.  The cost of these operations is a function of the size of the queue which is in turn a function of the choices made by the subroutine $\init(\MM)$.

\begin{definition}\label{def:initialColoring}
A \emph{leaf assignment} $\chi$ of a binary tree $T$ assigns \emph{two} distinct leaves to each internal vertex $v$,
one from the left and one from the right subtree of $v$ (or only one leaf if $v$ has only one child). 
The subroutine $\init(\MM)$ naturally defines a leaf assignment to $\cJ(\MM)$ (which 
is a rooted binary tree) according to the priority queue for each up-star we put a given vertex in.  
Call this the \emph{initial coloring} of the vertices in $\cJ(\MM)$, and denote it by $\chi$.  

For a vertex $v$ in $\cJ(\MM)$, let $L(v)$ denote the set of leaves of the subtree rooted at $v$, and let $A(v)$ denote the set of 
ancestors of $v$, i.e.\ the vertices on the $v$ to root path. For a vertex $v\in \cJ(\MM)$, and an initial coloring $\chi$, 
we use $H_v$ to denote the \emph{heap} at $v$. Formally, $H_v = \{u | u \in A(v), \chi(u) \cap L(v) \neq \emptyset\}$, i.e.\
the set of ancestors colored by some leaf in $L(v)$.
\end{definition}

Given the above technical definition, the proof of the following lemma is straightforward, though long and so has been moved to \Sec{runTimeProof}.

\begin{lemma}\label{lem:runTimeUpper}
Let $\MM$ be a simplicial complex with $t$ critical points.  
The running time of $\build(\MM)$ is $O(N+t\alpha(t) + \sum_{v \in \cJ(\MM)} \log |H_v|)$, where $H_v$ is defined by an initial coloring.
\end{lemma}
}


\InSoCGVer{Our main result, } \Thm{main-corr}, is an easy corollary of the above lemma. Specifically, consider a critical point $v$ of the initial input 
complex.  By \Thm{jordan} this vertex appears in exactly one of the pieces output by $\rain$.  As in the \Thm{main-corr} statement, let $\ell_v$ 
denote the length of the longest directed path passing through $v$ in the contour tree of the input complex, and let $\ell_v'$ 
denote the longest directed path passing through $v$ in the join tree of the piece containing $v$.
By \Thm{surgery}, ignoring non-dominant extrema introduced from cutting (whose cost can be charged to a corresponding saddle), the join tree on each piece output by $\rain$ is isomorphic to some 
connected subgraph of the contour tree of the input complex, and hence $\ell_v'\leq \ell_v$.  
Moreover, $|H_v|$ only counts vertices in a $v$ to root path and so trivially $|H_v|\leq \ell_v'$, implying \Thm{main-corr}.

Note that there is fair amount of slack in this argument as $|H_v|$ may be significantly smaller than $\ell_v'$.  
This slack allows for the more refined upper and lower bounds mentioned in \Sec{more-refined}.  Quantifying this slack however is quite challenging, 
and requires a significantly more sophisticated analysis involving path decompositions, which is the subject of \Sec{pathDecomp} and \Sec{lb}. 



\newcommand{\leafAssignPathDecomp}{
\section{Leaf assignments and path decompositions}
\label{sec:pathDecomp}

In this section, we set up a framework to analyze the time taken to compute a join tree $\cJ(\MM)$ (see \InSoCGVer{\Def{criticalJoin}}\InNotSoCGVer{\Def{join}}).
We adopt all notation already defined in \Sec{runTime}.  From here forward we will often assume binary trees are full binary trees 
(this assumption simplifies the presentation but is not necessary).

Let $\chi$ be some fixed leaf assignment to a rooted binary tree $T$, which in turn fixes all the heaps $H_v$. 
We choose a special path decomposition that is best defined as a subset of edges in $T$ such that each internal
vertex has degree at most $2$. This naturally gives a path decomposition.
For each internal vertex $v\in T$, add the edge from $v$ to $\arg \max_{v_l, v_r} \{|H_{v_l}|, |H_{v_r}|\}$ 
 where $v_l$ and $v_r$ are the children of $v$ (if $|H_{v_l}|=|H_{v_r}|$ then pick one arbitrarily).
This is called the \emph{maximum} path decomposition, denoted by $\pmax(T)$.

Our main goal in this section is to prove the following theorem.  We use $|p|$ to denote the number of vertices in $p$.

\begin{theorem} \label{thm:runtime} $\sum_{v \in T} \log |H_v| = O(\sum_{p\in \pmax(T)} |p| \log |p|)$.
\end{theorem}

We conclude this section in \Sec{implications} by showing that proving this theorem implies our main result \Thm{main-alg}.

\subsection{Shrubs, tall paths, and short paths}
\label{sec:shrubs}

The paths in $P(T)$ naturally define a tree\footnote{Please excuse the 
overloading of the term 'tree', it is the most natural term to use here.} of their own.  Specifically, in the original 
tree $T$ contract each path down to its root.  Call the resulting tree the \emph{shrub} of $T$ corresponding to 
the path decomposition $P(T)$. Abusing notation, we simply use $P(T)$ to denote the shrub.
As a result, we use terms like `parent', `child', `sibling', etc. for paths as well.
The shrub gives a handle on the heaps of a path.
We use $b(p)$ to denote the \emph{base} of the path, which is vertex in $p$ closest
to root of $T$. We use $\ell(p)$ to denote the leaf in $p$.
We use $H_p$ to denote the $H_{b(p)}$.

\begin{lemma}
\label{lem:adjacent}
 Let $p$ be any path in $P(T)$ and let $\{q_1, \dots q_k\}$ be the children on $p$.
Then $H_{\ell(p)} + \sum_{i=1}^k |H_{q_i}| \leq |H_p|+2|p|$. 
\end{lemma}

\begin{proof} For convenience, denote $H_i = H_{q_i}$ and $H_0 = H_{\ell(p)}$.
Consider $v \in \bigcup_{i} H_i$ that lies below $b(p)$ in $T$. 
Note that such a vertex has only one of its two colors in $L(b(p))$.  
Since the colors tracked by $H_i$ and $H_j$ for $i\neq j$ are disjoint, 
such a vertex can appear in only one of the $H_i$'s.
On the other hand, a vertex $u\in p$ can appear in more than one $H_i$, 
but since any vertex has exactly two colors it can appear in at most two 
such heaps.  Hence, $\sum_i |H_i| \leq |H_p| + 2|p|$.
\ignore{
 First observe that for all $i$, $r_{i}$ lies above $r_p$ on some root to leaf path. 
 Let $H_i'$ be the subset of $H_i$ that lies below $r_p$.  Observe that $H_i'\subseteq H_p$, as $L(T_{r_i})\subseteq L(T_{r_p})$.
  
 Now for any $i\neq j$, $H_i \cap H_j = \emptyset$, as $r_{i}$ and $r_{j}$ are not on the same root to leaf path (i.e.\ $L(T_{r_i}) \cap L(T_{r_j}) = \emptyset$).
 In particular, $(H_i\setminus H_i') \cap (H_j \setminus H_j') = \emptyset$, and so $\sum_{i=0}^k |H_i\setminus H_i'| \leq |p|$ 
 as $H_i\setminus H_i'$ is precisely the subset of $H_i$ that lies on $p$.
 Moreover, for any $i\neq j$, $ H_i' \cap H_j' = \emptyset$, and so since $H_i'\subseteq H_p$ we have $\sum_{i=0}^k |H_i'|\leq |H_p|$.
 Putting these together gives
 \[
  \sum_{i=0}^k |H_i| = \sum_{i=0}^k |H_i'| + \sum_{i=0}^k |H_i\setminus H_i'| \leq |H_p|+|p|.
 \]
}
\end{proof}

We wish to prove $\sum_{v\in T} \log |H_v| = O(\sum_{p\in P} |p| \log |p|)$. 
The simplest approach is to prove $\forall p\in P$, $\sum_{v\in p} \log |H_v|= O(|p|\log|p|)$.
This is unfortunately not true, which is why we divide paths into two categories.

\begin{definition}
 For $p\in P(T)$, $p$ is \emph{short} if $|p| < \sqrt{|H_{p}|}/100$, and \emph{tall} otherwise.
\end{definition}

\ignore{
\begin{observation}
\label{obs:decrease}
 Let $v$ be a vertex in $T$ and let $w$ be its parent.  Then $|H_w| \geq |H_v| -1$, as $L(T_v) \subseteq L(T_w)$ and the path from 
 $w$ to the root has one less vertex than the path from $v$ to the root.
 In particular, we have the following more general property.
 Let $v$ and $u$ be any two vertices in the same root to leaf path of $T$, such that $v$ is a descendant of $u$.  
 Then $|H_v| \leq |H_u| + d(u,v)$. 
\end{observation}

The following can be thought of as a generalization of the above observation, and will be useful in later sections.
}

The following lemma demonstrates that tall paths can ``pay'' for themselves.

\begin{lemma}
\label{lem:pathBounds} If $p$ is tall, $\sum_{v\in p} \log |H_v| = O(|p| \log |p|)$.
If $p$ is short, $\sum_{v\in p} \log |H_v| = O(|H_p| \log |H_p|)$.
\end{lemma}
\begin{proof} \ignore{
 By $\Obs{decrease}$ we know that for any vertex $v\in p$, $|H_v|\leq |H_{p}| + |p|$ (as $v$ is a descendant of $r_p$ along $p$).
}
For $v \in p$, $|H_v|\leq |H_{p}| + |p|$ (as $v$ is a descendant of $b(p)$ along $p$).
Hence, $\sum_{v\in p} \log |H_v| \leq \sum_{v\in p} \log(|H_{p}| + |p|)  = |p| \log (|H_{p}| + |p|)$.
 If $p$ is a tall path, then $|p| \log (|H_{p}| + |p|) = O(|p| \log |p|)$. If $p$ is short, then 
 $|p| \log (|H_{p}| + |p|) = O(|p| \log |H_{p}| )$. For short paths, $|p| = O(|H_p|)$.
\end{proof}

\ignore{
For a short path $p$ we can think of the quantity $|p|(\log |H_p| - \log |p|)$ as the excess of $p$, i.e.\ what was not ``paid for'' locally by $p$.
Our eventual goal will be to show that by choosing the right path decomposition, this excess 
can be recharged to tall paths. In the next couple sections we introduce the appropriate path decomposition and prove 
some useful facts about it.  Then given these facts, in \Sec{} we are able to make the recharging argument.
}

There are some short paths that we can also ``pay" for.  Consider any short path $p$ in the shrub.  We will refer to the 
\emph{tall support chain} of $p$ as the tall ancestors of $p$ in the shrub which have a path to $p$ which does not use any 
short path (i.e.\ it is a chain of paths adjacent to $p$).

\begin{definition} \label{def:support} A short path $p$ is \emph{supported} if
at least $|H_p|/100$ vertices $v$ in $H_p$ lie in paths in the tall support chain of $p$.
%
%
\end{definition}

Let $\cL$ be the set of short paths, $\cL'$ be the set of supported short paths, and $\cH$ be the set of tall paths
given by $\pmax(T)$.
We now construct the shrub of unsupported short paths. Consider $p \in \cL \setminus \cL'$,
and traverse the chain of ancestors from $p$. Eventually, we must reach another short path $q$.
(If not, we have reached the root $r$ of $\pmax(T)$. Hence, $p$ is supported.)
Insert edge from $p$ to $q$, so $q$ is the parent of $p$ in $\cU$. This construction leads
to the shrub forest of $\cL \setminus \cL'$, where all the roots are supported short paths, 
and the remaining nodes are the unsupported short paths.

Most of the work goes into proving the following technical lemma.

\begin{lemma} \label{lem:cu-root} Let $\cU$ denote a connected component (shrub) in the shrub forest of $\cL \setminus \cL'$ and let $r$
be the root of $\cU$. (i) For any $v \in p$ such that $p \in \cU$,
$|H_v| = O(|H_r|)$. (ii) $\sum_{p \in \cU} |p| = O(|H_r|)$.
\end{lemma}

We split the remaining argument into two subsections. We first prove \Thm{runtime} from \Lem{cu-root},
which involves routine calculations. Then we prove \Lem{cu-root}, where the interesting work happens.

\subsection{Proving \Thm{runtime}} \label{sec:thm-runtime}

We split the summation into tall, short, and unsupported short paths.
\begin{eqnarray*} 
\sum_{p \in \cL} \sum_{v \in p} \log |H_v| & = & \sum_{p \in \cL \setminus \cL'} \sum_{v \in p} \log |H_v| + \sum_{p \in \cL'} \sum_{v \in p} \log |H_v| + \sum_{p \in \cH} \sum_{v \in p} \log |H_v|
\end{eqnarray*}
The last term can be bounded by $O(\sum_{p \in \pmax(T)} |p|\log |p|)$, by \Lem{pathBounds}.
The second term can be bounded by $O(\sum_{p \in \cL'} |H_p| \log |H_p|)$, by \Lem{pathBounds} again.
The following claim shows that this in turn is at most the last term.


\begin{claim} \label{clm:above} $\sum_{p \in \cL'} |H_p| \log |H_p| = O(\sum_{q \in \cH} \sum_{v \in q} \log |H_v|)$.
\end{claim}

\begin{proof} 
Pick $p \in \cL'$. 
%
%
As we traverse the tall support chain of $p$, there
are at least $|H_p|/100$ vertices of $H_p$ that lie in these paths. These are encountered in
a fixed order. Let $H'_p$ be the first $|H_p|/200$ of these vertices. When $v \in H'_p$
is encountered, there are $|H_p|/200$ vertices of $H_p$ not yet encountered. Hence,
$|H_v| \geq |H_p|/200$. Hence, $|H_p|\log |H_p| = O(\sum_{v \in H'_p} \log |H_v|)$.
Since all the vertices lie in tall paths,
we can write this as $O(\sum_{q \in \cH} \sum_{v \in H'_p \cap q} \log |H_v|)$.
Summing over all $p$, the expression is $\sum_{q \in \cH} \sum_{p \in \cL'} \sum_{v \in H'_p \cap q} \log |H_v|$.

Consider any $v \in H'_p$.  Let $S$ be the set of paths $\widetilde{p}\in \cL'$ such that 
$v\in H'_{\widetilde{p}}$.  We now show $|S|\leq 2$ (i.e.\ it contains at most one path other than $p$).
First observe that any two paths in $S$ must be unrelated (i.e.\ $S$ is an anti-chain), 
since paths which have an ancestor-descendant relationship have disjoint tall support chains.  
However, any vertex $v$ receives exactly one color from each of its two subtrees (in $T$), and therefore $|S|\leq 2$ since any two 
paths which share descendant leaves in $T$ (i.e.\ their heaps are tracking the same color) must have an ancestor-descendant relationship.

In other words, any $\log |H_v|$ appears at most twice in the above triple summation.
Hence, we can bound it by $O(\sum_{q \in \cH} \sum_{v \in q} \log |H_v|)$.
%
%
\end{proof}

The first term (unsupported short paths) can be charged to the second term (supported short paths).
This is where the critical \Lem{cu-root} plays a role.

\begin{claim} \label{clm:first} $\sum_{p \in \cL \setminus \cL'} \sum_{v \in p} \log |H_v| = O(\sum_{p \in \cL'} |H_p| \log |H_p|)$.
\end{claim}

\begin{proof} Let $\cU$ denote a connected component of the shrub forest.
We have $\sum_{p \in \cL \setminus \cL'} \sum_{v \in p} \log |H_v| \allowbreak \leq \sum_{\cU} \sum_{p \in \cU} \sum_{v \in p} \log |H_v|$.
By \Lem{cu-root}, $|H_v| = O(|H_r|)$, where $r$ is the root of $\cU$.
Furthermore, $\sum_{p \in \cU} |p| = O(|H_r|)$. 
We have $\sum_{p \in \cU} \sum_{v \in p} \log |H_v| = O((\log |H_r|) \sum_{p \in \cU} |p|) = O(|H_r|\log |H_r|)$.
We sum this over all $\cU$ in the shrub forest, and note that roots in the shrub forest are supported short paths.
\end{proof}

\subsection{Proving \Lem{cu-root}: the root is everything in $\cU$} \label{sec:weight}

\Lem{cu-root} asserts the root $r$ in $\cU$ pretty much encompasses all sizes and heaps in $\cU$.
We will work with the \emph{reduced} heap $\redH_p$. This is the subset of vertices of $H_p$ that 
do not appear on the tall support chain of $p$.
By definition, for any unsupported short path
(hence, any non-root $p \in \cU$), $|\redH_p| \geq 99|H_p|/100$. 
We begin with a key property, which is where the construction of $\pmax(T)$ enters the picture.

\begin{lemma}
\label{lem:geometric}
 Let $q$ be the child of some path $p$ in $\cU$, then $|H_p|\geq \frac{3}{2}|H_q|$. 
 Moreover, if $p\neq r(\cU)$, then $|\redH_p|\geq \frac{3}{2}|\redH_q|$. 
\end{lemma}
\begin{proof} Let $h(q)$ denote the tall path that is a child of $p$ in $\pmax(T)$,
and an ancestor of $q$. If no such tall path exists, then by construction $p$ is the parent 
of $q$ in $\pmax(T)$, and the following argument will go through by setting $h(q)=q$.

The chain of ancestors from $q$ to $h(q)$ consists only of tall paths.
Since $q$ is unsupported, these paths contain at most $|H_q|/100$ vertices of $H_q$. 
Thus, $|H_{h(q)}| \geq 99|H_q|/100$.

Consider the base
of $h(q)$, which is a node $w$ in $T$. Let $v$ denote the sibling of $w$ in $T$.
Their parent is called $u$. Note that both $u$ and $v$ are nodes in the path $p$.
Now, the decomposition $\pmax(T)$ put $u$ and $v$ in the same path $p$. This
implies $|H_v| \geq |H_w|$. Since $|H_u| \geq |H_v| + |H_w| - 2$,
$|H_u| \geq 2|H_w| - 2$. Let $b$ be the base of $p$.
We have $|H_p| = |H_b| \geq |H_u| - |p| \geq 2|H_w| - |p| - 2$. Since $p$ is a short path, $|p| < \sqrt{|H_p|}/100$.
Applying this bound, we get $|H_p| \geq (2-\delta)|H_w|$ (for a small constant $\delta > 0$).
Since $w$ is the base of $h(q)$, $H_w = H_{h(q)}$. We apply the bound $|H_{h(q)}| \geq 99|H_q|/100$
to get $|H_p| \geq 197|H_q|/100$, implying the first part of the lemma.  For the second part, observe that 
if $p\neq r(\cU)$, then $p$ is unsupported and so $|\redH_p| \geq 99|H_p|/100$, and therefore the second part follows since $|H_q| \geq |\redH_q|$.
%
%
\end{proof}

This immediately proves part (i) of \Lem{cu-root}. Part (ii) requires much more work.

We define a \emph{residue} $R_p$ for each $p \in \cU$. Suppose $p$ has children
$q_1, q_2, \ldots, q_k$ in $\cU$. Then $R_p = |\redH_p| - \sum_i |\redH_{q_i}|$. 
By definition, $|\redH_p| = \sum_{q \in \cU_p} R_p$. 
Note that $R_p$ can be negative. 
Now, define $R^+_p = \max(R_p,0)$, and set $W_p = \sum_{q \in \cU_p} R^+_p$.
Observe that $W_p \geq |\redH_p|$. We also get an approximate converse.

\begin{claim} \label{clm:loss} For any path $p\in \cU$, $|\redH_p| \geq W_p - 2\sum_{q \in \cU_p} |q|$.
\end{claim}

\begin{proof} We write $W_p - |\redH_p| = \sum_{q \in \cU_p} R^+_q - R_q$ $= -\sum_{q \in \cU_p: R_q < 0} R_q$.
Consider $q \in \cU_p$ and denote the children in $\cU_p$ by $q'_1, q'_2, \ldots$.
Note that $R_q$ is negative exactly when $|\redH_q| < \sum_i |\redH_{q'_i}|$.
Traverse $\pmax(T)$ from $q'_i$ to $q$.
Other than $q$, all other nodes encountered are in the tall support chain of $q'_i$ and hence
do not affect its reduced heap.
The vertices of $\redH_{q'_i}$ that are deleted
are exactly those present in the path $q$. 
Any vertex in $q$ can be deleted from at most two of the reduced heaps (of the children of $q$ in $\cU_p$), since 
theses reduced heaps do not have an ancestor-descendant relationship.
Therefore when $R_q$ is negative, it is at most by $2|q|$.
We sum over all $q$ to complete the proof.
\end{proof}

\ignore{Let $D_q = H_q \setminus \bigcup_i H_{q'_i}$. In other words, as we traverse $\cU$,
$D_q$ is the set of vertices ``added" at $q$. Suppose we traverse $\pmax(T)$ from $q'_i$
to $q$. All vertices of $H_{q'_i}$ that lie on these paths are removed.

 and traverse $\pmax(T)$ up to $p$. Suppose we encounter
the paths $q_0 = q, q_1, q_2, \ldots, q_k = p$. The vertices in the residue $R_q$ that 
are removed from $H_p$ are exactly the vertices that lie in these paths. 
We split this contribution into tall and short paths, observing that the latter
are all unsupported. Hence, we can (crudely) upper bound the number of vertices
removed by $\sum_{q \in \cU_p} |q| + \sum_{q \in \cU_p} \sum_{h \in \cH} |h \cap q|$.
Each vertex can be removed from at most two distinct $R_q$'s, completing the proof.
}

The main challenge of the entire proof is bounding the sum of path lengths, which is done next.
We stress that the all the previous work is mostly the setup for this claim.

\begin{claim} \label{clm:charge} Fix any path $p \in \cU\setminus\{r(\cU)\}$. Suppose for any $q,q' \in \cU_p$
where $q$ is a parent of $q'$ in $\cU_p$, $W_q \geq (4/3)W_{q'}$. Then $\sum_{q \in \cU_p} |q| \leq W_p/20$.
\end{claim}

\begin{proof} Since $q$ is an unsupported short path, $|q| < \sqrt{|H_q|}/100 \leq \sqrt{|\redH_q|}/99 \leq \sqrt{|W_q|}/99$.
We prove that $\sum_{q \in \cU_p} \sqrt{|W_q|}/99 \leq W_p/20$
by a charge redistribution scheme. Assume that each $q \in \cU_p^-$
starts with $\sqrt{W_q}/99$ units of charge. We redistribute this charge over all
nodes in $\cU_q$, and then calculate the total charge. For $q \in \cU_p$, spread
its charge to all nodes in $\cU_q$ proportional to $R^+$ values. In other words,
give $(\sqrt{W_q}/99)\cdot(R^+_{q'}/W_q)$ units of charge to each $q' \in \cU_q$.

After the redistribution, let us compute the charge deposited at $q$. Every ancestor in $\cU_p^-$ 
$q = a_0, a_1, a_2, \ldots, a_k$ contributes to the charge at $q$. The charge is expressed in the following
equation. We use the assumption that $W_{a_i} \geq (4/3) W_{a_{i-1}}$ and hence $W_{a_i} \geq (4/3)^i W_{a_0} \geq (4/3)^i$, 
as $a_0$ is an unsupported short path and hence $W_{a_0}\geq 1$. 
$$ ({R^+_q}/99)\sum_{a_i} 1/\sqrt{W_{a_i}} \leq ({R^+_q}/99)\sum_{a_i} (3/4)^{i/2} \leq R^+_q/20 $$
The total charge is $\sum_{q \in \cU_p} R^+_p/20 = W_p/20$. 
\end{proof}

\begin{corollary} \label{cor:chargeCor}
 Let $r$ be the root of $\cU$, and suppose that for any paths $q,q' \in \cU\setminus \{r\}$, where $q$ is a parent of $q'$ in $\cU$, $W_q \geq (4/3)W_{q'}$.
 Then $\sum_{p \in \cU} |p| \leq W_r/20+|r|$.
\end{corollary}
\begin{proof}
 Let $c_1, \dots, c_m$ be the children of $r$ in $\cU$.  By definition, $W_r = \sum_i W_{c_i}+R^+_r \geq \sum_i W_{c_i}$.
 By \Clm{charge}, for each $c_i$ we have $W_{c_i}/20\geq \sum_{p \in \cU_{c_i}} |p|$.  Combining these to facts yields the claim.
\end{proof}

%
%

We wrap it all up by proving part (ii) of \Lem{cu-root}.

\begin{claim} $\sum_{p \in \cU} |p| \leq |H_{r(\cU)}|/10$.
\end{claim}

%

\begin{proof} We use $r$ for $r(\cU)$. 
Suppose $W_q \geq (4/3)W_{q'}$ (for any choice in $\cU\setminus\{r\}$ of $q$ parent of $q'$), 
then by \Cor{chargeCor}, $\sum_{p \in \cU} |p| \leq W_{r}/20+|r|$.
By \Clm{loss}, $|\redH_r| \geq W_r - 2\sum_{p \in \cU} |p|$, and so combining these inequalities gives,
\[
 \sum_{p \in \cU} |p|\leq \frac{10}{9} \pth{|\redH_r|/20+|r|} \leq 
 \frac{10}{9}\pth{|H_r|/20+\sqrt{|H_r|}/100} \leq |H_r|/10.
\]

We now prove that for any $q$ parent of $q'$ (other than $r$), $W_q \geq (4/3)W_{q'}$.
Suppose not. Let $p, p'$ be the counterexample furthest from the root,
where $p$ is the parent of $p'$. 
Note that for $q$ and child $q'$ in $\cU_{p'}$, $W_q \geq (4/3)W_{q'}$.
We will apply \Clm{charge} for $\cU_{p'}$ to deduce that $\sum_{q \in \cU_{p'}} |q| \leq W_{p'}/20$.
Combining this with \Clm{loss} gives, $|\redH_{p'}| \geq 19W_{p'}/20$. By \Lem{geometric}, $|\redH_p| \geq (3/2)|\redH_{p'}|$.
Noting that $W_p \geq |\redH_p|$, we deduce that $W_p \geq (4/3)W_{p'}$.
Hence, $p, p'$ is not a counterexample, and more generally, there is no counterexample.
That completes the whole proof.
\end{proof}

\subsection{Our Main Result}
\label{sec:implications}

We now show that \Thm{runtime} allows us to upper bound the running time for our join tree and contour 
tree algorithms in terms of path decompositions.

\begin{theorem}
Let $f:\MM \to \RR$ be the linear interpolant over distinct valued vertices, where the join tree $\cJ(\MM)$ has maximum degree $3$.
There is an algorithm to compute the join tree whose running time is $O(\sum_{p \in \pmax(\cJ)} |p|\log |p| + t\alpha(t) + N)$.
\end{theorem}
\begin{proof}
By \Thm{correct} we know that $\build(\MM)$ correctly outputs $\cJ(\MM)$, and by \Lem{runTimeUpper}
we know this takes $O(\sum_{v \in \cJ(\MM)} \log |H_v| + t\alpha(t) + N)$ time, where the $H_v$ values are determined as in \Def{initialColoring}.
Therefore by \Thm{runtime}, $\build(\MM)$ takes $O(\sum_{p \in \pmax(\cJ)} |p|\log |p| + t\alpha(t) + N)$ 
time to correctly compute $\cJ(\MM)$.
\end{proof}

This result for join trees easily implies our main result, \Thm{main-alg}, which we now restate and prove.

\begin{theorem}
Let $f:\MM \to \RR$ be the linear interpolant over distinct valued vertices, where the contour tree $\cC=\cC(\MM)$ has maximum degree $3$. 
There is an algorithm to compute $\cC$ whose running time is $O(\sum_{p \in P(\cC)} |p|\log |p| + t\alpha(t) + N)$,
where $P(T)$ is a specific path decomposition (constructed implicitly by the algorithm).
\end{theorem}
\begin{proof}
First, lets review the various pieces of our algorithm. 
On a given input simplicial complex, we first break it into extremum dominant pieces using $\rain(\MM)$ (and in $O(|\MM|)$ time by \Thm{rain-time}).  
Specifically, \Lem{rain-1} proves that the output of $\rain(\MM)$ is a set of extremum dominant pieces, 
$\MM_1, \dots, \MM_k$, and \Clm{rain-reeb} shows that given the contour trees, $\cC(\MM_1), \dots, \cC(\MM_k)$, 
the full contour tree, $\cC(\MM)$, can be constructed (in $O(|\MM|)$ time).

Now one of the key observations was that for extremum dominant manifolds, computing the contour tree is roughly the same as 
computing the join tree.  Specifically, \Thm{contour-tree} implies that given $\jc(\MM_i)$ , 
we can obtain $\cC(\MM_i)$ by simply sticking on the non-dominant minima at their respective splits (which can easily be done in linear time).
Remark~\ref{rem:joinAndCriticalJoin} implies that $\jc(\MM_i)$ is trivially obtained from the $\cJ(\MM_i)$, and 
 by the above theorem we know $\cJ(\MM_i)$ can be computed in $O(\sum_{p \in \pmax(\cJ(\MM_i))} |p|\log |p| + t_i\alpha(t_i) + N_i)$
(where $t_i$ and $N_i$ are the number of critical points and faces when restricted to $\MM_i$).

At this point we can now see what the path decomposition referenced in theorem statement should be.  It is just the union of all the maximum 
path decomposition across the extremum dominant pieces, $\pmax(\cC(\MM)) = \cup_{i=1}^k \pmax(\cJ(\MM_i))$.
Since all procedures besides computing the join trees take linear time in the size of the input complex, we can therefore 
compute the contour tree in time
\InSoCGVer{
\begin{align*}
\hspace{.8in} O\pth{N+ \sum_{i=1}^k \pth{\sum_{p \in \pmax(\cJ(\MM_i))} |p|\log |p|} + t_i\alpha(t_i) + N_i}& \\
=\break O\pth{\pth{\sum_{p \in \pmax(\cC(\MM))} |p|\log |p|} + t\alpha(t) + N}&
\end{align*}
}
\InNotSoCGVer{
\[
O\pth{N+ \sum_{i=1}^k \pth{\sum_{p \in \pmax(\cJ(\MM_i))} |p|\log |p|} + t_i\alpha(t_i) + N_i}
=\break O\pth{\pth{\sum_{p \in \pmax(\cC(\MM))} |p|\log |p|} + t\alpha(t) + N}
\]
}




\end{proof}
}



\InNotSoCGVer{
\leafAssignPathDecomp
}



\newcommand{\lowerBoundbyPathDecomp}{
\section{Lower Bound by Path Decomposition}
\label{sec:lb}
We first prove a lower bound for join trees, and then generalize to contour trees.
Note that the form of the theorem statements in this section differ from \Thm{main-lb}, 
as they are stated directly in terms of path decompositions.  \Thm{main-lb} 
is an immediate corollary of the final theorem of this section, \Thm{path-main-lb}.

\subsection{Join Trees}

We focus on terrains, so $d=2$.
Consider any path decomposition $P$ of a valid join tree (i.e.\ any rooted binary tree).
When we say ``compute the join tree", we require the join tree to
be labeled with the corresponding vertices of the terrain.

\begin{figure}[h]\centering
    \includegraphics[width=.18\linewidth]{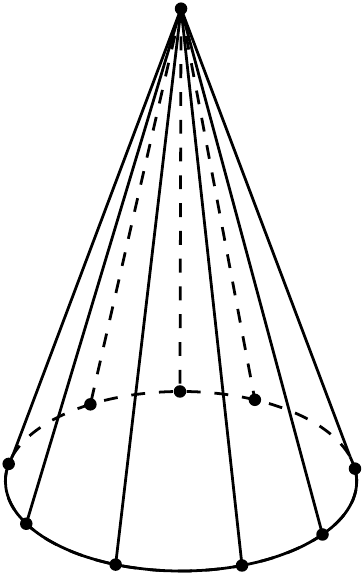}\hspace{2cm}
    \includegraphics[width=.3\linewidth]{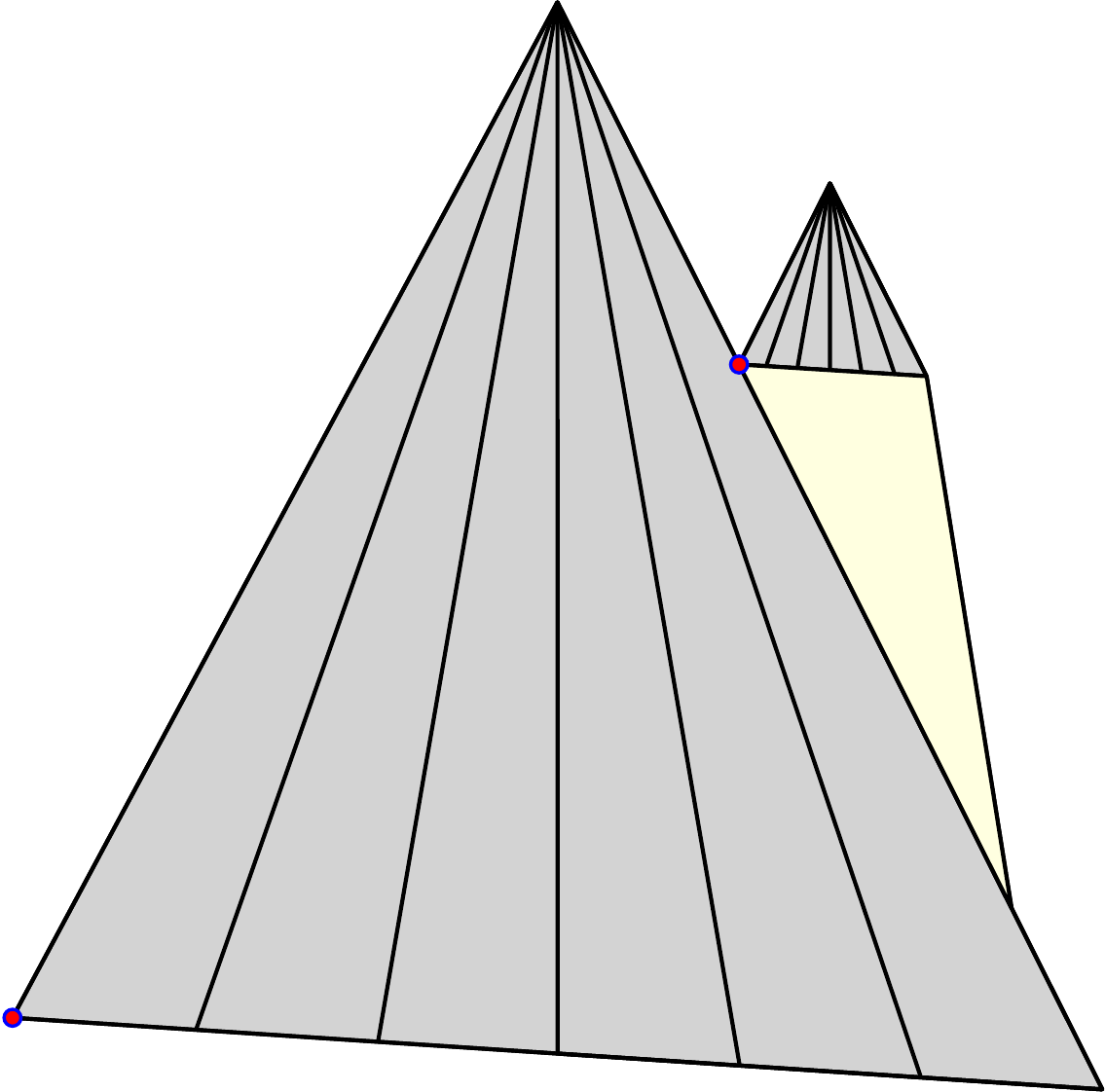}
    \caption{Left: angled view of a tent / Right: a parent and child tent put together}
    \label{fig:diamond}
\end{figure}

\begin{lemma}
\label{lem:pathDecomp}
Fix any path decomposition $P$. There is a family of terrains, $\bF_P$, all with the same triangulation, 
such that $|\bF_P| = \Pi_{p_i\in P} (|p_i|-1)!$, and no two terrains in $\bF_P$ define the same join tree.
\end{lemma}

\begin{proof}
 
 We describe the basic building block of these terrains, which corresponds to a fixed path $p\in P$.  
 Informally, a \emph{tent} is an upside down cone with $m$ triangular faces (see \Fig{diamond}). 
 Construct a slightly tilted cycle 
 of length $m$ with the two antipodal points at heights $1$ and $0$. These are called the anchor and trap of the tent, respectively.
 The remaining $m-2$ vertices are evenly spread around the cycle and heights decrease monotonically when going from the anchor to the trap.
 Next, create an apex vertex at some appropriately large height, and add an edge to each vertex in the cycle.

 Now we describe how to attach two different tents.  In this process, we glue the base of a scaled down ``child'' tent on to a triangular cone face of the larger  ``parent'' tent (see \Fig{diamond}).
 Specifically, the anchor of the child tent is attached directly to a face of the parent tent at some height $h$. The remainder of the base of the child cone is then 
 extended down (at a slight angle) until it hits the face of the parent.
 
 The full terrain is obtained by repeatedly gluing tents. For each path $p_i\in P$, we create a tent of size $|p_i|+1$.
 The two faces adjacent to the anchor are always empty, and the remaining faces are for gluing on other tents.
 (Note that tents have size $|p_i|+1$ since $|p_i|-1$ faces represent the joins of $p_i$, the apex represents the leaf, 
 and we need two empty faces next to the anchor.)
 Now we glue together tents of different paths in the same way the paths are connected in the shrub $\pathTree$ (see \Sec{shrubs}).  
 Specially, for two paths $p,q\in P$ where 
 $p$ is the parent of $q$ in $\pathTree$, we glue $q$ onto a face of the tent for $p$ as described above.
 (Naturally for this construction to work, tents for a given path will be scaled down relative to the size of the tent of their parent).
 By varying the heights of the gluing, we get the family of terrains.
 
 Observe now that the only saddle points in this construction are the anchor points.  Moreover, the only maxima are the apexes of the 
 tents.  We create a global boundary minimum by setting the vertices at the base of the tent representing the root of $\pathTree$ all to the same height (and there are no other minima).  
 Therefore, the saddles on a given tent will appear contiguously on a root to leaf path in the join tree of the terrain, where the leaf corresponds to the maximum of the tent 
 (since all these saddles have a direct line of sight to this apex).  In particular, 
 this implies that, regardless of the heights assigned to the anchors, the join tree has a path decomposition whose corresponding shrub is equivalent to $\pathTree$.  
 
 There is a valid instance of this described construction for any relative ordering of the heights of the saddles on a given tent.  
 In particular, there are $(|p_i|-1)!$ possible orderings of the heights of the saddles on the tent for $p_i$, and hence $\Pi_{p_i\in P} (|p_i|-1)!$ possible 
 terrains we can build.  Each one of these functions will result in a different (labeled) join tree. All saddles on a given tent will appear in sorted order in the join tree.
 So, any permutation of the heights on a given tent corresponds to a permutation of the vertices along a path in $P$.
\end{proof}

Two path decompositions $P_1$ and $P_2$ (of potentially different complexes and/or height functions) are equivalent 
if: there is a 1-1 correspondence between the sizes of the constituent paths, and the shrubs are isomorphic.
%
%

\begin{lemma}
\label{lem:cost}
For all $\MM \in \bF_P$, the total number of heap operations performed by $\build(\MM)$ is $O(\sum_{p\in P} |p|\log|p|)$.
\end{lemma}
\begin{proof}
The primary ``non-determinism" of the algorithm is the initial painting constructed by $\init(\MM)$. 
We show that regardless of how paint spilling is done, the number of heap operations is bounded as above.

 Consider an arbitrary order of the initial paint spilling over the surface.
 Consider any join on a face of some tent, which is the anchor point of some connecting child tent.  
 The join has two up-stars, each of which has exactly one edge. Each edge connects to a maximum
 and must be colored by that maximum.
 Hence, the two colors touching this join (according to \Def{initialColoring})
 are the colors of the apexes of the child and parent tent.  
%
 
 Take any join $v$, with two children $w_1$ and $w_2$. Suppose $w_1$ and $v$ belong to the same path
 in the decomposition. The key is that any color from a maximum in the subtree at $w_2$ cannot touch
 any ancestor of $v$. This subtree is exactly the join tree of the child tent attached at $v$.
 The base of this tent is completely contained in a face of the parent tent. So all colors
 from the child ``drain off" to the base of the parent, and do not touch any joins on the parent tent.
 
 Hence, $|H_v|$ is at most the size of the path in $P$ containing $v$. By \Lem{runTimeUpper}, the total number of heap operations
 is at most $\sum_v \log |H_v|$, completing the proof.
 
%
\end{proof}

The following is the equivalent of \Thm{main-lb} for join trees, and immediately follows from
the previous lemmas.

\begin{theorem}\label{thm:joinLB}
Consider a rooted tree $T$ and an arbitrary path decomposition $P$ of $T$.
There is a family $\bF_P$ of terrains such that any algebraic decision
tree computing the join tree\footnote{Note that for the referenced family of terrains, the join tree and contour tree are equivalent} 
(on $\bF_P$) requires $\Omega(\sum_{p \in P} |p|\log |p|)$ time.
Furthermore, our algorithm makes $O(\sum_{p \in P} |p|\log |p|)$ comparisons on all these instances.
%
\end{theorem}

\begin{proof}
The proof is a basic entropy argument. Any algebraic decision tree that is correct on all of $\bF_P$
must distinguish all inputs in this family. By Stirling's approximation, the depth of this tree is $\Omega(\sum_{p_i\in P} |p_i|\log|p_i|)$.
\Lem{cost} completes the proof.
%
%
\end{proof}

\subsection{Contour Trees}
We first generalize previous terms to the case of contour trees.
In this section $T$ will denote an arbitrary contour tree with every internal vertex of degree $3$.

For simplicity we now restrict our attention to path decompositions 
consistent with the raining procedure described in \Sec{rain} 
(more general decompositions can work, but it is not needed for our purposes).

\begin{definition}
\label{def:path2} A path decomposition, $P(T)$, is called \emph{rain consistent} if its paths can be obtained as follows.
Perform an downward BFS from an arbitrary maximum $v$ in $T$, and mark all vertices encountered.  
Now recursively run a directional BFS from all vertices adjacent to the current marked set.  
Specifically, for each BFS run, make it an downward BFS if it is at an odd height in the recursion tree and upward otherwise.  

This procedure partitions the vertex set into disjoint rooted subtrees of $T$, based on which BFS marked a vertex.  
For each such subtree, now take any partition of the vertices into leaf paths.\footnote{Note that the subtree  
of the initial vertex is rooted at a maximum.  For simplicity we require that the path this vertex belongs to 
also contains a minimum.}
\end{definition}

The following is analogous to \Lem{pathDecomp}, and in particular uses it as a subroutine.

\begin{lemma}
\label{lem:pathDecomp2}
Let $P$ be any rain consistent path decomposition of some contour tree. 
There is a family of terrains, $\bF_P$, all with the same triangulation, 
such that the size of $\bF_P$ is $\Pi_{p_i\in P} (|p_i|-1)!$, and no two terrains in $\bF_P$ define the same contour tree.
\end{lemma}
\begin{proof}
 As $P$ is rain consistent, the paths can be partitioned into sets $P_1, \dots, P_k$, where $P_i$ is the set of all paths with 
 vertices from a given BFS, as described in \Def{path2}.  Specifically, let $T_i$ be the subtree of $T$ corresponding to $P_i$ and 
 let $r_i$ be the root vertex of this subtree.  
 Note that the $P_i$ sets naturally define a tree where $P_i$ is the parent of $P_j$ if $r_i$ (i.e.\ the root of $T_i$) is adjacent to a vertex in $P_j$.  
 
 As the set $P_i$ is a path decomposition of a rooted binary tree $T_i$, the terrain construction of \Lem{pathDecomp} for $P_i$ is well defined.  
 Actually the only difference is that here the rooted tree is not a full binary tree, and so some of the (non-achor adjacent) faces of the constructed tents will be blank.  
 Specifically, these blank faces correspond to the adjacent children of $P_i$, and they tell us how to connect the terrains of the different $P_i$'s.
 
 So for each $P_i$ construct a terrain as described in \Lem{pathDecomp}.  Now each $T_i$ is (roughly speaking) a join or a split tree, depending on 
 whether the BFS which produced it was an upward or downward BFS, respectively.  As the construction in \Lem{pathDecomp} was for join trees, each terrain 
 we constructed for a $P_i$ which came from a split tree, must be flipped upside down.
 Now we must described how to glue the terrains together.
  
 \begin{figure}[h]\centering
    \includegraphics[width=.28\linewidth]{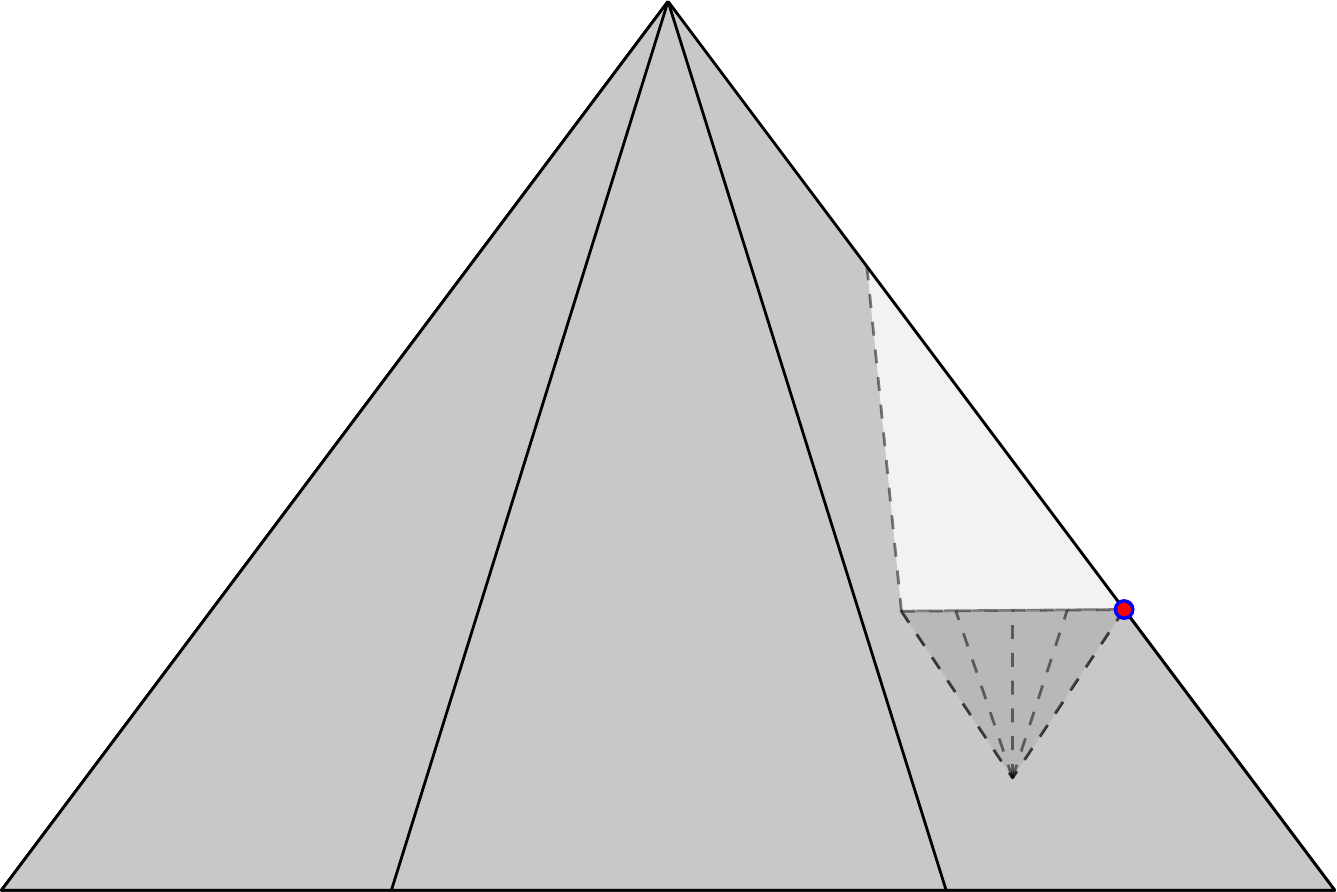}
    \caption{A child tent attached to a parent tent with opposite orientation.}
    \label{fig:diamond2}
\end{figure} 
  
 By construction, the tents corresponding to the paths in $P_i$ are connected 
 into a tree structure (i.e.\ corresponding to the shrub of $P_i$).  Therefore the bottoms of all these tents are covered except for the one corresponding 
 to the path containing the root $r_i$.  
 If $r_i$ corresponds to the initial maximum that the rain consistent path decomposition was defined from, then this will be flat and corresponds to the global outer face.
 Otherwise, $P_i$ has some parent $P_j$ in which case we connect the bottom of the tent for $r_i$ to a free face of a tent in the construction for $P_j$, specifically, the face corresponding to the 
 vertex in $T$ which $r_i$ is adjacent to. This gluing is done in the same manner as in \Lem{pathDecomp}, attaching the anchor for the root of $P_i$ directly the corresponding face of $P_j$, except that now $P_i$ and $P_j$ have opposite orientations.  See \Fig{diamond2}.
 
 Just as in \Lem{pathDecomp} we now have one fixed terrain structure, such that each different relative ordering of the heights of the join and split vertices 
 on each tent produces a surface with a distinct contour tree.  The specific bound on the size of $\bF_P$, defining these distinct contour trees, follows by applying the 
 bound from \Lem{pathDecomp} to each $P_i$.
\end{proof}

\begin{lemma}
For all $\MM \in \bF_P$, the number of heap operations is $\Theta(\sum_{p\in P} |p|\log |p|)$
\end{lemma}
\begin{proof}
 This lemma follows immediately from \Lem{cost}.  The heap operations can be partitioned into the
 operations performed in each $P_i$.
 Apply \Lem{cost} to each of the $P_i$ separately and take the sum.
\end{proof}

We now restate \Thm{main-lb}, which follows immediately from an entropy argument, analogous to \Thm{joinLB}.

\begin{theorem}\label{thm:path-main-lb}
Consider any rain consistent path decomposition $P$. There exists a family $\bF_P$ of terrains ($d=2$) with the following properties.
Any contour tree algorithm makes $\Omega(\sum_{p \in P} |p|\log |p|)$ comparisons in the worst case over $\bF_P$.
Furthermore, for any terrain in $\bF_P$, our algorithm makes $O(\sum_{p \in P} |p|\log |p|)$ comparisons.
\end{theorem}

\begin{Remark}
 Note that for the terrains described in this section, the number of critical points is within a constant factor of 
 the total number of vertices.  In particular, for this family of terrains, all previous algorithms required $\Omega(n\log n)$ time.
\end{Remark}
}


\InNotSoCGVer{
\lowerBoundbyPathDecomp
}

\myparagraph{Acknowledgements.}
We thank Hsien-Chih Chang, Jeff Erickson, and Yusu Wang for numerous useful discussions. This work is
supported by the Laboratory Directed Research and Development
(LDRD) program of Sandia National Laboratories. Sandia National
Laboratories is a multi-program laboratory managed and operated
by Sandia Corporation, a wholly owned subsidiary of Lockheed
Martin Corporation, for the U.S. Department of Energy's National
Nuclear Security Administration under contract DE-AC04-94AL85000.

\InNotSoCGVer{
\bibliographystyle{alpha}
}

\InSoCGVer{
\bibliographystyle{plain}
}

\bibliography{contour}

\InSoCGVer{
\appendix



\section{Some technical remarks}
\label{sec:techmarks}
The following is a continuation from the end of \Sec{basics}.

\technicalRemarks


\section{Proof of Surgery}
\label{sec:proofOfSurgery}
\surgeryBody
That is, Theorem~\ref{thm:jordan}, which is proven below.
\proofofDisconnected


\section{Join and Split Trees of Extremum Dominant Manifolds}
\label{sec:jstedm}
In this section we prove Theorem~\ref{thm:contour-tree} (of \Sec{extreme}), which states the equivalence 
of the contour tree on extremum dominant manifolds and the critical join tree of Definition~\ref{def:criticalJoin}.
This proof requires some definitions and notation, specially
the notions of \emph{join} and \emph{split} trees, as given by~\cite{csa-cctad-00}.
Conventionally, all edges are directed from higher to lower function value. 

\joinTreeSplitTreeMerge

We are now ready to prove Theorem~\ref{thm:contour-tree}.
\proofofContourJoinEquiv


\section{Proof of \Lem{runTimeUpper}}
\label{sec:runTimeProof}
\proofofRunTime


\leafAssignPathDecomp


\lowerBoundbyPathDecomp


}

%
%
%
%

\appendix

\end{document}